\documentclass[twocolumn]{autart}

\usepackage{graphicx}
\usepackage{amsmath,amssymb,bm}
\usepackage{todonotes}
\usepackage{enumitem}
\usepackage{url}
\usepackage{booktabs}

\newcommand{\real}{\ensuremath{\mathbb{R}}}

\newcommand{\smat}[1]{\ensuremath{\left[\begin{smallmatrix}#1\end{smallmatrix}\right]}}
\newcommand{\bmat}[1]{\ensuremath{\begin{bmatrix}#1\end{bmatrix}}}

\newcommand{\tu}[1]{{\textup{#1}}}
\newcommand{\mb}[1]{\mathbf{#1}}
\newcommand{\dVx}{\ensuremath{\frac{\partial V}{\partial x}(x)}}
\newcommand{\dVxS}{\ensuremath{\tfrac{\partial V}{\partial x}(x)}}
\newcommand{\Zc}{\ensuremath{Z_{\tu{c}}}}

\DeclareMathSymbol{\cdoT}{\mathord}{symbols}{"01}

\DeclareMathOperator{\diag}{diag}

\newtheorem{remark}{Remark}
\newtheorem{lemma}{Lemma}
\newenvironment{proof}{\textit{Proof.}}{\hfill \hspace*{0.01pt}  \hfill$\square$}
\newenvironment{proofof}[1]{\textit{Proof of #1.}}{\hfill \hspace*{0.01pt}  \hfill$\square$}

\newtheorem{corollary}{Corollary}
\newtheorem{theorem}{Theorem}
\newtheorem{proposition}{Proposition}
\newtheorem{assumption}{Assumption}

\newcommand\thisis{arxiv}
\newcommand{\submission}[1]{\ifthenelse{\equal{\thisis}{submission}}{#1}{}}
\newcommand{\arxiv}[1]{\ifthenelse{\equal{\thisis}{arxiv}}{#1}{}}

\allowdisplaybreaks
\graphicspath{{./figures/}}

\begin{document}

\begin{frontmatter}
\title{
Data-driven control via Petersen's lemma
\thanksref{footnoteinfo}}%
\thanks[footnoteinfo]{This paper was not presented at any IFAC 
meeting. Corresponding author A.~Bisoffi. This research is partially supported by a Marie Skłodowska-Curie
COFUND grant, no. 754315 and by NWO, project no. 15472.}
\author[gron]{Andrea Bisoffi}\ead{a.bisoffi@rug.nl},
\author[gron]{Claudio De Persis}\ead{c.de.persis@rug.nl},
\author[fi]{Pietro Tesi}\ead{pietro.tesi@unifi.it} 
\address[gron]{ENTEG and J.C. Willems Center for Systems and Control, University of Groningen, 9747 AG Groningen, The Netherlands}
\address[fi]{DINFO, University of Florence, 50139 Florence, Italy}       

\begin{abstract}
We address the problem of designing a stabilizing closed-loop control law directly from input and state measurements collected in an open-loop experiment. 
In the presence of noise in data, we have that a set of dynamics could have generated the collected data and we need the designed controller to stabilize such set of data-consistent dynamics robustly.
For this problem of data-driven control with noisy data, we advocate the use of a popular tool from robust control, Petersen's lemma. 
In the cases of data generated by linear and polynomial systems, we conveniently express the uncertainty captured in the set of data-consistent dynamics through a matrix ellipsoid, and we show that a specific form of this matrix ellipsoid makes it possible to apply Petersen's lemma to all of the mentioned cases. 
In this way, we obtain necessary and sufficient conditions for data-driven stabilization of linear systems through a linear matrix inequality. 
The matrix ellipsoid representation enables insights and interpretations of the designed control laws. 
In the same way, we also obtain sufficient conditions for data-driven stabilization of polynomial systems through (convex) sum-of-squares programs. 
The findings are illustrated numerically.
\end{abstract}
\end{frontmatter}

\section{Introduction}
\label{sec:intro}

\subsection*{Motivation and Petersen's lemma}

Data-driven control design is a relevant methodology to tune controllers whenever modelling from first principles is challenging, the model parameters are possibly redundant and cannot be unambiguously identified through suitable experiments, while (possibly large) datasets can be obtained from the process to be controlled. 
Thanks to the technological trend that measurements are increasingly easier to access and retrieve, using data to directly design controllers has witnessed a renewed surge in interest in recent years \cite{Dai2018cdc,Coulson2018,recht_annual,Baggio2019,depersis2020tac,berberich2019robust,vanwaarde2020noisy}.

These recent developments have been drawing results from classical areas of control theory such as behavioural theory \cite{Coulson2018,depersis2020tac,dorfler2021bridging}, set-membership system identification, and robust control \cite{berberich2019robust,Dai2018cdc}. 
A pivotal role in many of these developments has been played by the so-called fundamental lemma by Willems et al.~\cite[Thm.~1]{willems2005note}; qualitatively speaking, this result shows that for a linear system, controllability and persistence of excitation ensure that its representation through dynamical matrices $(A,B)$ is equivalent to a representation through a finite-length trajectory; however, such trajectory is assumed not to be affected by noise. 
The inevitable presence of noise in data, then, prevents from representing equivalently the actual system and induces rather a set of systems that could have generated the noisy data for a given bound on the noise, i.e., the set of systems consistent with data. 
This set, which we call $\mathcal{C}$, plays a central role since
control design must therefore target all systems in $\mathcal{C}$, which are indistinguishable from each other based on data.

A natural way to address this uncertainty induced by noisy data is via robust control tools: e.g., system level synthesis \cite{Dean2018journ,Xue2020datadriven,anderson2019system},
Young's inequality for matrices \cite{depersis2020tac}, matrix generalizations of the S-procedure \cite{vanwaarde2020noisy,ferizbegovic2019learning}, Farkas's lemma \cite{Dai2018cdc,dai2021semi}, linear fractional transformations \cite{berberich2019robust,berberich2019combining}.
We advocate here the use of another robust control tool for data-driven control, Petersen's lemma \cite{petersen1987stabilization,petersen1986riccati}. 
This lemma, whose strict and nonstrict versions we report later in Facts~\ref{fact:petersen-ext} and \ref{fact:petersen-nonstrict}, can be seen as a matrix elimination method since, instead of verifying for all matrices bounded in norm a certain inequality, one can \emph{equivalently} verify another inequality where such matrices do not appear.
The utility of Petersen's lemma in the realm of robust control has been featured in \cite{shcherbakov2008extensions,khlebnikov2008petersen,ji2016extension}. Petersen's lemma underpins the data-based results of this work, as detailed next.

\subsection*{Contributions}

Our main contributions are the following.
\emph{(C1)} We bring Petersen's lemma to the attention as a powerful 
tool for data-driven control.
\emph{(C2)} For linear systems, we provide by it necessary and sufficient conditions for quadratic stabilization, which are
alternative to those in~\cite{vanwaarde2020noisy}. 
These conditions take the convenient form of linear matrix inequalities.
\emph{(C3)} We give several insights on the design conditions and, in particular, establish connections with certainty equivalence and robust indirect control, which have been extensively investigated for stochastic noise models, e.g., \cite{Dean2018journ,ferizbegovic2019learning,treven2021learning}.
\emph{(C4)} For polynomial systems, we obtain new sufficient conditions for data-driven control with respect to \cite{dai2021semi,Guo2020}.
These conditions are tractably relaxed into alternate (convex) sum-of-squares programs.

\subsection*{Relations with the literature}

We assume an upper bound on the norm of the noise sequence, 
which is the so-called \emph{unknown-but-bounded} noise paradigm \cite{hjalmarsson1993discussion}.
This makes our approach different from those considering stochastic noise descriptions \cite{Dean2018journ,recht_annual,ferizbegovic2019learning,treven2021learning} 
and similar in nature to set-membership identification and control \cite{fogel1979system,milanese2004set,tanaskovic2017data}.
The use of robust control tools to counteract the uncertainty induced by unknown-but-bounded noise is quite natural and has been pursued in 
\cite{berberich2019robust,berberich2019combining,dai2021semi,depersis2020tac,Guo2020,vanwaarde2020noisy}.
Next, we compare with these works referring to our aforementioned contributions \emph{(C1)}-\emph{(C4)}.\newline
\emph{(C1)} The use of Petersen's lemma differentiates our approach from those in \cite{depersis2020tac,berberich2019robust,berberich2019combining,vanwaarde2020noisy}, which also address data-driven stabilization of linear systems (besides $\mathcal{H}_2$, $\mathcal{H}_\infty$ or quadratic performance).
In \cite{Bisoffi2020bilinear}, we used Petersen's lemma only as a sufficient condition \cite[Fact~1]{Bisoffi2020bilinear} to obtain a data-driven controller for structurally different bilinear systems. \newline
\emph{(C2)} For linear systems in discrete time, \cite{vanwaarde2020noisy} provided necessary and sufficient conditions for data-based stabilization as we do here. 
The differences are illustrated in detail in Section~\ref{sec:comparison}. 
In a nutshell, here we operate under an easy-to-enforce condition stemming from persistence of excitation instead of under a generalized Slater condition, and the former (but not the latter) 
can be seamlessly satisfied also in the relevant special case of ideal data (i.e., without noise).\newline
\emph{(C3)} For the considered noise setting, the uncertainty set $\mathcal C$ consists in a \emph{matrix ellipsoid}, whose
center is the (ordinary) least-squares estimate of the system dynamics, and whose
size depends on the noise bound. This justifies why \emph{certainty-equivalence} control can be expected to work well
in regimes of small uncertainty (small noise), which agrees with recent works on performance of certainty-equivalence control for linear quadratic control
\cite{mania2019certainty,dorfler2021certainty}. 
On the other hand, this also explains why robust design 
is generally needed to have stability guarantees, which is 
also the main idea behind the \emph{robust indirect} control approaches \cite{Dean2018journ,ferizbegovic2019learning,treven2021learning}
under a stochastic noise description. 
On a related note, we introduced the notion of matrix ellipsoid in \cite{Bisoffi2021tradeoffs}, which had however a quite different focus and research question.
\newline
\emph{(C4)} Data-driven control of polynomial systems was proposed also in~\cite{Guo2020,dai2021semi}. 
As in~\cite{Guo2020}, we use Lyapunov methods to obtain sufficient conditions for data-based global asymptotic stabilization. 
Whereas \cite{Guo2020} parametrizes the Lyapunov function in a specific way, the present data-based conditions parallel naturally the classical model-based ones in \cite{khalil2002nonlinear} since they correspond to enforcing those model-based conditions (through Petersen's lemma) for all systems consistent with data, which leads to succinct derivations.
Due to this natural parallel, the present approach appears to be extendible with appropriate modifications to other cases where Lyapunov(-like) conditions occur, as we do for \emph{local} asymptotic stabilization in Corollary~\ref{cor:extens to LAS}.
On the other hand, \cite{dai2021semi} follows a radically different approach. 
Instead of Lyapunov functions, it uses \emph{density functions} by Rantzer \cite{rantzer2001dual} to give a necessary and sufficient condition for data-based stabilization, which however needs to be relaxed into a quadratically-constrained quadratic program through sum of squares and then into a semidefinite program through moment-based techniques for tractability.

\subsection*{Structure}

In Section~\ref{sec:start}, we recall Petersen's lemma, formulate the problem and derive some properties of the set $\mathcal{C}$. 
In Section~\ref{sec:data-driv lin:main results} we provide our main results for linear systems and comment the results in Section~\ref{sec:lin discussion}. 
In Section~\ref{sec:data-driven polynomial} we provide our main result for polynomial systems. 
All results are exemplified numerically in Section~\ref{sec:num ex}.

\section{Preliminaries and problem setting}
\label{sec:start}

\subsection{Notation and Petersen's lemma}
\label{sec:notation and petersen}

For a vector $a$, $|a|$ denotes its 2-norm.
For a matrix $A$, $\|A\|$ denotes its induced 2-norm, which is equivalent to the largest singular value of $A$; moreover, for a scalar $a>0$, $\| A \| \le a$ if and only if $A^\top A \preceq a^2 I$ where $I$ is the identity matrix.
For matrices $A$, $B$ and $C$ of compatible dimensions, we abbreviate $A B  C (AB)^\top$ to $A B \cdoT C[\star]^\top$, where the dot in the second expression clarifies unambiguously that $AB$ are the terms to be transposed. 
For matrices $A=A^\top$, $B$, $C=C^\top$, we also abbreviate the symmetric matrix $\smat{A & B \\ B^\top & C}$ as $\smat{A & B \\ \star & C}$ or $\smat{A & \star \\ B^\top & C}$. 
For a positive semidefinite matrix $A$, $A^{1/2}$ denotes the unique positive semidefinite root of $A$.
For a matrix $A$, $A^\dagger$ denotes the Moore-Penrose generalized inverse of $A$, which is uniquely determined by certain axioms \cite[p.~453, 7.3.P7]{horn2013matrix}.
For positive integers $n$, $r$ and the set $\mathcal{P}$ of polynomials $p \colon \real^n \to \real$ (resp., the set $\mathcal{P}_{\tu{m}}$ of matrix polynomials $P \colon \real^n \to \real^{r\times r}$), the set $\mathcal{S} \subset \mathcal{P}$ (resp. $\mathcal{S}_{\tu{m}} \subset \mathcal{P}_{\tu{m}}$) denotes the set of sum-of-squares polynomials (resp., the set of sum-of-squares matrix polynomials) in the variable $x \in \real^n$; see \cite{chesi2010lmi} and references therein for more details on these and other sum-of-squares notions.

Petersen's lemma is the essential tool we use to address data-driven control design. 
First, we present in the next fact a version where inequalities are strict.

\begin{fact}[Strict Petersen's lemma]
\label{fact:petersen-ext}
Consider matrices $\mb{C} \in \real^{n \times n}$, $\mb{E} \in \real^{n \times p}$, $\overline{\mb{F}}  \in \real^{q \times q}$, $\mb{G} \in \real^{q \times n}$ with $\mb{C}=\mb{C}^\top$ and $\overline{\mb{F}} = \overline{\mb{F}}^\top \succeq 0$, and let $\mathcal{F}$ be
\begin{equation}
\label{set cal F}
\mathcal{F}:=\{\mb{F} \in \real^{p \times q}\colon \mb{F}^\top \mb{F} \preceq \overline{\mb{F}} \}. 
\end{equation}
Then,
\begin{subequations}
\begin{equation}
\label{petersen-ext:for all}
\mb{C}+\mb{E} \mb{F} \mb{G} + \mb{G}^\top \mb{F}^\top \mb{E}^\top \prec 0 \text { for all } \mb{F} \in \mathcal{F}
\end{equation}
if and only if there exists $\lambda>0$ such that
\begin{equation}
\label{petersen-ext:multipl}
\mb{C}+\lambda \mb{E} \mb{E}^\top +\lambda^{-1} \mb{G}^\top  \overline{\mb{F}} \mb{G} \prec 0.
\end{equation}
\end{subequations}
\end{fact}

For $\overline{\mb{F}} =I$, one obtains the original version by I.~R. Petersen in \cite{petersen1986riccati,petersen1987stabilization}, and the version in Fact~\ref{fact:petersen-ext} proposes a slight extension where the bound $\overline{\mb{F}}$ is any positive semidefinite matrix.
For this version, then, we give the proof in the appendix for completeness. 
Although one could prove Fact~\ref{fact:petersen-ext} with S-procedure arguments as some authors do for \emph{nonstrict} versions \cite{shcherbakov2008extensions,khlebnikov2008petersen}, we follow the original proof strategy of~\cite{petersen1986riccati,petersen1987stabilization}.

Second, we present in the next fact a version of Petersen's lemma where inequalities are nonstrict.
\begin{fact}[\hspace{0.5pt}Nonstrict\hspace{0.6pt}Petersen's\hspace{0.6pt}lemma]
\label{fact:petersen-nonstrict}%
Consider matrices $\mb{C} \in \real^{n \times n}$, $\mb{E} \in \real^{n \times p}$, $\overline{\mb{F}} \in \real^{q \times q}$, $\mb{G} \in \real^{q \times n}$ with $\mb{C}=\mb{C}^\top$ and $\overline{\mb{F}}  = \overline{\mb{F}}^\top \succeq 0$, and let $\mathcal{F}$ be defined as in~\eqref{set cal F}. 
Suppose additionally $\mb{E} \neq 0$, $\overline{\mb{F}} \succ 0$ and $\mb{G} \neq 0$. Then,
\begin{subequations}
\begin{equation}
\label{petersen-nonstrict:for all}
\mb{C}+\mb{E} \mb{F} \mb{G} + \mb{G}^\top \mb{F}^\top \mb{E}^\top \preceq 0 \text { for all } \mb{F} \in \mathcal{F}
\end{equation}
if and only if there exists $\lambda>0$ such that
\begin{equation}
\label{petersen-nonstrict:multipl}
\mb{C}+\lambda \mb{E} \mb{E}^\top +\lambda^{-1} \mb{G}^\top \overline{\mb{F}} \mb{G} \preceq 0.
\end{equation}
\end{subequations}
Moreover, \eqref{petersen-nonstrict:multipl} implies \eqref{petersen-nonstrict:for all} \emph{without} the assumption $\mb{E} \neq 0$, $\overline{\mb{F}} \succ 0$ and $\mb{G} \neq 0$.
\end{fact}

For $\overline{\mb{F}}=I$, one obtains precisely the nonstrict versions of Petersen's lemma in \cite[\S 2]{shcherbakov2008extensions} and \cite[\S 2]{khlebnikov2008petersen};
for completeness we then report the proof of Fact~\ref{fact:petersen-nonstrict} \submission{in \cite{arXivVersion}}\arxiv{in the appendix}.
The additional assumption with respect to Fact~\ref{fact:petersen-ext} (i.e., $\mb{E} \neq 0$, $\overline{\mb{F}} \succ 0$ and $\mb{G} \neq 0$) is due to having \emph{nonstrict} inequalities and is needed to obtain the specific form \eqref{petersen-nonstrict:multipl}\submission{, see \cite{arXivVersion}.}\arxiv{. Indeed, with $\overline{\mb{F}} \succ 0$, Fact~\ref{fact:petersen-nonstrict} would hold also for either $\mb{E} \neq 0$, $\mb{G} = 0$ or $\mb{E} = 0$, $\mb{G} \neq 0$ if $\lambda \ge 0$ is allowed (instead of $\lambda >0$) and \eqref{petersen-nonstrict:multipl} is replaced by, respectively, either $\smat{\mb{C} +\lambda \mb{E} \mb{E}^\top & \mb{G}^\top \overline{\mb{F}}^{1/2}\\
\overline{\mb{F}}^{1/2}\mb{G} & -\lambda I} \preceq 0$ or $\smat{\mb{C} +\lambda \mb{G}^\top \overline{\mb{F}} \mb{G} & \mb{E}\\
 \mb{E}^\top & -\lambda I} \preceq 0$;
Fact~\ref{fact:petersen-nonstrict} is trivially true if $\mb{E} = 0$ and $\mb{G} = 0$.}

\subsection{Problem formulation}
\label{sec:probl form}

Consider a discrete-time linear time-invariant system
\begin{equation}
\label{sys}
x^+ =  A_\star x + B_\star u + d
\end{equation}
where $x \in \real^n$ is the state, $u \in \real^m$ is the input, $d \in \real^n$ is a disturbance, and the matrices $A_\star$ and $B_\star$ are \emph{unknown} to us.
At the same time and with the same meaning for the quantities $x$, $u$ and $d$, consider the continuous-time linear time-invariant system
\begin{equation}
\label{sys-ct}
\dot x =  A_\star x + B_\star u + d.
\end{equation}
The modifications required for the continuous-time case are limited, and this allows us to treat it in parallel to the discrete-time case.
Instead of relying on model knowledge given by $A_\star$ and $B_\star$, we perform an experiment on the system by applying an input sequence $u(t_0)$, $u(t_1)$, \dots, $u(t_{T-1})$ of $T$ samples, so that by \eqref{sys}/\eqref{sys-ct}
\[
x(t_{i+1})/\dot{x}(t_i) = A_\star x(t_i) + B_\star u(t_i) + d(t_i) .
\]
We measure the state response $x(t_0)$, $x(t_1)$, \dots, $x(t_{T-1})$, and, in discrete time, the shifted state response $x(t_1)$, $x(t_2)$, \dots, $x(t_{T})$ or, in continuous time, the state-derivative response $\dot x(t_0)$, $\dot x(t_1)$, \dots, $\dot x(t_{T-1})$.
The disturbance sequence $d(t_0)$, $d(t_1)$, \dots, $d(t_{T-1})$ affects the evolution of the system and is \emph{unknown} to us, hence data are \emph{noisy}. 
We collect the noisy data in the matrices
\begin{subequations}
\label{data}
\begin{align}
\hspace*{-2.2mm} U_0 & :=\bmat{u(t_0) & u(t_1) & \cdots & u(t_{T-1})}\\
\hspace*{-2.2mm} X_0 & :=\bmat{x(t_0) & x(t_1) & \cdots & x(t_{T-1})} \\
\hspace*{-2.2mm} X_1 & :=\bmat{x(t_1) & x(t_2) & \cdots & x(t_T)} \text{ in discrete time, or}\\
\hspace*{-2.2mm} X_1 & :=\bmat{\dot x(t_0) & \dot x(t_1) & \cdots & \dot x(t_{T-1})} \text{ in continuous time} .
\end{align}
\end{subequations}%
We can also arrange the unknown disturbance sequence as
$D_0 :=\bmat{d(t_0) & d(t_1) & \cdots & d(t_{T-1})}$, so that
$D_0$ and data in~\eqref{data} satisfy
\begin{equation}
\label{data rel D0}
X_1 = A_\star X_0 + B_\star U_0 + D_0
\end{equation}
since \eqref{sys} (in discrete time) or \eqref{sys-ct} (in continuous time) is the underlying data generation mechanism. 
In the former case, we have $t_0$, $t_1$, \dots, $t_T$ equal to, respectively, $0$, $1$, \dots, $T$; in the latter case, $t_0$, $t_1$, \dots, $t_{T-1}$ are sampled periodically at $0$, $T_{\tu{s}}$, \dots, $(T-1) \cdot T_{\tu{s}}$ for some sampling time $T_{\tu{s}}$, although this is not necessary. 

We operate under a certain disturbance model. 
Specifically, we assume that the disturbance sequence $D_0$ has bounded energy, i.e., $D_0 \in \mathcal D$ where, for some matrix $\Delta$,
\begin{equation}
\label{set D}
\mathcal D := \{ 
D \in \real^{n \times T}: 
DD^\top \preceq \Delta \Delta^\top \}.
\end{equation}
As we said, $D_0$ is unknown to us and the only a-priori knowledge on it is given by the set $\mathcal{D}$, and in particular the knowledge of the positive semidefinite bound $\Delta \Delta^\top$.
This disturbance model enforces an energy bound on the disturbance since it constrains the whole disturbance sequence, unlike an instantaneous disturbance bound \cite{Bisoffi2021tradeoffs}. 
Energy bounds are used in \cite{depersis2020tac,berberich2019robust,vanwaarde2020noisy,berberich2019combining} and many other works.
In fact, model \eqref{set D} is quite general as it can capture signal-to-noise-ratio conditions \cite{depersis2020tac}, over-approximate instantaneous bounds \cite{Bisoffi2021tradeoffs}, and can also be used to have probabilistic bounds for Gaussian noise \cite{dpt2021Aut}.

With data \eqref{data} and set $\mathcal{D}$ in~\eqref{set D}, we  introduce the set $\mathcal{C}$ of matrices \emph{consistent with data}
\begin{equation}
\label{set C}
\mathcal C := \{ 
(A,B):  X_1 = A X_0 + B U_0 + D,  D \in \mathcal D \},
\end{equation}
i.e., the set of all pairs $(A,B)$ of dynamical matrices that could  generate data $X_1$, $X_0$ and $U_0$ based on \eqref{sys} or \eqref{sys-ct} while keeping the disturbance sequence in the set $\mathcal{D}$. 
This is elucidated by comparing \eqref{set C} with the similar \eqref{data rel D0}. 
We note that $D_0 \in \mathcal{D}$ is precisely equivalent to $(A_\star,B_\star) \in \mathcal{C}$. 

\begin{remark}
In the language of set-membership identification 
\cite{milanese2004set}, we have two \emph{prior assumptions}, the first one on the class of dynamical systems \eqref{sys} or \eqref{sys-ct} and the second one on the noise \eqref{set D}. 
The set $\mathcal{C}$ in \eqref{set C} corresponds to the \emph{feasible systems set} \cite[Def.~1]{milanese2004set}. 
We noted that $(A_\star,B_\star) \in \mathcal{C}$. 
This corresponds to \emph{validation of prior assumptions} \cite[Def.~2]{milanese2004set}.
\end{remark}

Our objective is to design a state feedback controller
\[
u = K x
\]
that makes the closed-loop matrix $A_\star + B_\star K$ Schur stable (i.e., all its eigenvalues have magnitude less than 1) in discrete time, or Hurwitz stable (i.e., all its eigenvalues have real part less than 0) in continuous time.
However, we lack the knowledge of $(A_\star, B_\star)$ and the disturbance $d$ induces uncertainty in data, which results into a set $\mathcal{C}$ of matrices consistent with data. 
Our objective becomes then to stabilize robustly all matrices $A + B K$ for $(A,B) \in \mathcal{C}$; in other words, in discrete time,
\begin{subequations}
\label{probl}
\begingroup%
\thinmuskip=0mu plus 4mu
\medmuskip=0mu plus 8mu
\thickmuskip=1mu plus 12mu
\begin{align}
\hspace*{-1mm}& \text{find} & & \hspace*{-2mm}K,~P=P^\top\succ 0 \label{probl:find} \\
\hspace*{-2mm}& \text{s.~t.} & &  \notag\\  
\hspace*{-2mm}& & &  
\hspace*{-7mm}(A+BK) P (A+BK)^\top - P \prec 0 \text{ for all }(A,B) \in \mathcal C. \label{probl:for all}
\end{align}
\endgroup
\end{subequations}
or, in continuous time,
\begin{subequations}
\label{probl-ct}
\begingroup%
\thinmuskip=0mu plus 4mu
\medmuskip=0mu plus 8mu
\thickmuskip=1mu plus 12mu
\begin{align}
\hspace*{-1mm}& \text{find} & & \hspace*{-2mm}K,~P=P^\top\succ 0 \label{probl-ct:find} \\
\hspace*{-1mm}& \text{s.~t.} & &  \notag\\  
\hspace*{-1mm}& & &  
\hspace*{-7mm}(A+BK) P + P (A+BK)^\top \prec 0 \text{ for all }(A,B) \in \mathcal C. \label{probl-ct:for all}
\end{align}
\endgroup
\end{subequations}
Both~\eqref{probl} and \eqref{probl-ct} are quadratic stabilization problems. 
Achieving the objective of robust stabilization of all matrices $A + B K$ for $(A,B) \in \mathcal{C}$ (hence, also of $A_\star + B_\star K$) guarantees bounded-input bounded-state stability of $x^+ =  (A_\star + B_\star K) x + d$ or $\dot x =  (A_\star + B_\star K) x + d$ by \cite[Thm.~9.5]{antsaklis2006linear}.

\subsection{Reformulations of set $\mathcal C$ and properties}
\label{sec:reform set C}

We perform some rearrangements of $\mathcal{C}$.
We substitute in \eqref{set C} the definition of set $\mathcal{D}$ in~\eqref{set D} and obtain
\begin{align*}
& \mathcal{C} = \Big\{ (A,B):  X_1 = A X_0 + B U_0 + D,D \in \real^{n \times T},\\
& \hspace*{38mm} \smat{I & D}\smat{-\Delta \Delta^\top & 0\\ 0 & I}\smat{I \\ D^\top} \preceq 0\Big\}. 
\end{align*}
In this expression we substitute $D = X_1 - A X_0 - B U_0$ in the matrix inequality and collect $\big[I~~A~~B\big]$ on the left and its transpose on the right of the matrix inequality; then, $\mathcal{C}$ rewrites equivalently as
\begin{align}
\label{set C:with A,B,C}
 \mathcal{C} & = \Big\{
(A,B)\colon
\begin{bmatrix}
I &  A &  B
\end{bmatrix}
\cdoT
\left[ \begin{array}{c|c}
\mb{C} & \mb{B}^\top \\ 
\hline
\mb{B} & \mb{A}
\end{array} \right]
[\star]^\top
\preceq 0
\Big\} \\
& = \big\{
\big[A~~B\big] = Z^\top \!\colon 
\mb{C} + \mb{B}^\top\! Z + Z^\top\! \mb{B} + Z^\top\! \mb{A} Z \preceq 0 \big\}\label{set C: with A,B,C expand}
\end{align}
where we rearrange $(A,B)$ as the matrix $\big[A~~B\big]$ from \eqref{set C:with A,B,C} to \eqref{set C: with A,B,C expand} and define
\begin{equation}
\label{A B C consist ellips}
\left[ \begin{array}{c|c}
\mb{C} & \mb{B}^\top \\ 
\hline
\mb{B} & \mb{A}
\end{array} \right]
:=
\left[ \begin{array}{c|c}
X_1X_1^\top-\Delta\Delta^\top & -X_1 \bmat{X_0\\U_0}^\top \\
\hline
\phantom{\rule{0.1pt}{25pt}} 
-\bmat{X_0\\U_0}X_1^\top & \bmat{X_0\\ U_0} \bmat{X_0\\ U_0}^\top
\end{array} \right].
\end{equation}

\begin{remark}
\label{remark:comparison dist model}
For given matrices $R=R^\top$, $Q=Q^\top \succ 0$, $S$, one can consider a disturbance model $\mathcal D^\prime := \{ D \in \real^{n \times T} \colon \smat{I & D} \smat{R & S^\top\\ S & Q} \smat{I \\ D^\top} \preceq 0 \}$ more general than $\mathcal{D}$ in~\eqref{set D}, as in~\cite{berberich2019robust,vanwaarde2020noisy,berberich2019combining}. 
With $\mathcal{D}^\prime$, one can still carry out the derivations for a set $\mathcal{C}^\prime$ similar to~\eqref{set C:with A,B,C}, with slightly different expressions of $\mb{A}^\prime$, $\mb{B}^\prime$, $\mb{C}^\prime$. 
Nonetheless, \eqref{set D} is general enough to capture interesting classes of noise, see the discussion after \eqref{set D}.
\end{remark}

We make the next assumption on matrix $\smat{X_0\\ U_0}$ in~\eqref{A B C consist ellips}.
\begin{assumption}
\label{ass:pers exc}
Matrix $\smat{X_0\\ U_0}$ has full row rank.
\end{assumption}

Assumption 1 is related to persistence of excitation as we illustrate in Section~\ref{sec:pers of exc}, and can be checked directly from data.
If its condition does not hold, it can typically be enforced by collecting more data points (i.e., adding more columns to $\smat{X_0\\ U_0}$).
An immediate consequence of Assumption~\ref{ass:pers exc} for $\mb{A}$ in~\eqref{A B C consist ellips} is $\mb{A} \succ 0$.

The set $\mathcal{C}$ in~\eqref{set C: with A,B,C expand} can be regarded as a \emph{matrix} ellipsoid, i.e., a natural extension of the standard (vector) ellipsoid \cite[p.~42]{boyd1994linear} with parameters $\mb{c} \in \real$, $\mb{b} \in \real^p$, $\mb{a} \in \real^{p \times p}$:
\begin{equation*}
\{
z \in \real^p \colon 
\mb{c} + \mb{b}^\top z + z^\top \mb{b} + z^\top \mb{a} z \le 0
\}.
\end{equation*}
In fact, if a scalar system with $n=m=1$ is considered, $\mathcal{C}$ reduces to a standard ellipsoid with $Z \in \real^2$. 
The interpretation of $\mathcal{C}$ as a matrix ellipsoid (introduced in \cite{Bisoffi2021tradeoffs} to compute a size for this set) proves useful here since it enables a final simple reformulation of $\mathcal{C}$ as
\begin{equation}
\label{set C: with A,Q,Zc}
\mathcal{C} = \big\{ \big[A~~B\big] = Z^\top \colon 
(Z - \Zc)^\top \mb{A} (Z - \Zc) \preceq \mb{Q}
\big\}
\end{equation}
where, by $\mb{A} \succ 0$ from Assumption~\ref{ass:pers exc}, we define
\begin{equation}
\label{Zc,Q}
\Zc := - \mb{A}^{-1} \mb{B}, \quad \mb{Q}:= \mb{B}^\top \mb{A}^{-1}\mb{B} - \mb{C},
\end{equation}
as can be verified by substituting \eqref{Zc,Q} into \eqref{set C: with A,Q,Zc} and expanding all products to obtain \eqref{set C: with A,B,C expand}.
We will further discuss later in Section~\ref{sec:least squares} the interpretation of some of the parameters $\mb{A}$, $\mb{B}$, $\mb{C}$, $\Zc$, $\mb{Q}$ of $\mathcal{C}$. 
The matrix-ellipsoid parametrizations of $\mathcal{C}$ in \eqref{set C: with A,B,C expand}, \eqref{set C: with A,Q,Zc} and \eqref{set E} (given below) are analogous to the parametrizations of a standard ellipsoid as, respectively, a quadratic form \cite[Eq.~(3.8)]{boyd1994linear}, as a center and shape matrix and as a linear transformation of a unit ball \cite[Eq.~(3.9)]{boyd1994linear}. 
We report the sign definiteness of $\mb{A}$ and $\mb{Q}$ in the next lemma.
\begin{lemma}
\label{lemma:sign of A,Q}
Under Assumption~\ref{ass:pers exc}, $\mb{A} \succ 0$ and $\mb{Q} \succeq 0$.
\end{lemma}
\begin{proof}
$\mb{A} \succ 0$ from Assumption~\ref{ass:pers exc} by \cite[Thm.~7.2.7(c)]{horn2013matrix}. 
As for $\mb{Q}$, $\mb{A} \succ 0$ allows defining
\[
\mb{Q}_{\tu{p}}
:= \bmat{X_0\\ U_0}^\top \Bigg( \bmat{X_0\\ U_0} \bmat{X_0\\ U_0}^\top \Bigg)^{-1} \bmat{X_0\\ U_0}.
\]
$\mb{Q}_{\tu{p}}$ is a projection matrix, i.e., $\mb{Q}_{\tu{p}}^2 = \mb{Q}_{\tu{p}}$ as one verifies immediately. 
Then, \eqref{Zc,Q} and \eqref{A B C consist ellips} yield
\begin{equation}
\label{Q alt}
\mb{Q}= X_1 \mb{Q}_{\tu{p}} X_1^\top - X_1 X_1^\top + \Delta \Delta^\top.
\end{equation}
Write \eqref{data rel D0} as $X_1 = \smat{A_\star & B_\star} \smat{X_0\\U_0} + D_0$; this expression and $\mb{Q}_{\tu{p}}$ being a projection matrix shows that $\mb{Q}_{\tu{p}} X_1^\top - X_1^\top = \mb{Q}_{\tu{p}} (\smat{X_0\\U_0}^\top \smat{A_\star & B_\star}^\top  + D_0^\top) -  (\smat{X_0\\U_0} ^\top \smat{A_\star & B_\star}^\top  + D_0^\top) = (\mb{Q}_{\tu{p}} - I) D_0^\top$. 
Using this expression in~\eqref{Q alt} yields
\begin{equation}
\label{Q simplest}
\mb{Q} = D_0 (\mb{Q}_{\tu{p}} - I) D_0^\top + \Delta \Delta^\top \succeq - D_0 D_0^\top + \Delta \Delta^\top \succeq 0
\end{equation}
since $\mb{Q}_{\tu{p}} \succeq 0$ and $D_0 \in \mathcal{D}$ (i.e., $D_0D_0^\top \preceq \Delta \Delta^\top$).
\end{proof}

From $\mb{A} \succ 0$, we have the next desirable property of $\mathcal{C}$.
\begin{lemma}
\label{lemma:boundedness of set C}
Under Assumption~\ref{ass:pers exc}, $\mathcal{C}$ is bounded with respect to any matrix norm.
\end{lemma}
\begin{proof}
Consider $\mathcal{C}$ in~\eqref{set C: with A,Q,Zc}, which is nonempty as $\Zc \in \mathcal{C}$. $Z \in \mathcal{C}$ if and only if for all $v \in \real^n$, $v^\top (Z-Z_{\tu{c}})^\top \mb{A} (Z-Z_{\tu{c}}) v \le  v^\top \mb{Q} v$. Denote $\lambda_{\min}(\mb{A})$ the minimum eigenvalue of (symmetric) $\mb{A}$.
By Lemma~\ref{lemma:sign of A,Q}, this implies
\begin{align*}
& \sqrt{\lambda_{\min}(\mb{A})} |(Z-\Zc) v| \le |\mb{Q}^{1/2} v| \text{ for all } v \colon |v|=1\\
& \implies \sqrt{\lambda_{\min}(\mb{A})} \sup_{|v|=1} |(Z-\Zc) v| \le \sup_{|v|=1}  |\mb{Q}^{1/2} v| \\
& \implies  \|Z-\Zc\| \le \lambda_{\min}(\mb{A})^{-1/2} \|\mb{Q}^{1/2} \| \\
& \implies  \|Z\| \le \| \Zc \| + \lambda_{\min}(\mb{A})^{-1/2} \|\mb{Q}^{1/2}\|
\end{align*}
where we used the definition of induced 2-norm and the reverse triangle inequality in the second and third implication, respectively. All quantities on the right hand side are finite, so each $Z \in \mathcal{C}$ has bounded 2-norm. 
Recall that any two matrix norms are equivalent \cite[p.~371]{horn2013matrix}, so for any given pair of matrix norms $\|\cdot\|_a$ and $\| \cdot \|_b$, there is a finite constant $C_{ab}>0$ such that $\|M\|_a \le C_{ab} \| M \|_b$ for all matrices $M$.
Hence, boundedness of $\mathcal{C}$ with respect to the induced 2-norm implies boundedness of $\mathcal{C}$ with respect to any other norm, as needed proving.
\end{proof}

\section{Data-driven control for linear systems}
\label{sec:data-driv lin:main results}

So far, we have rewritten the set $\mathcal{C}$ of dynamical matrices $(A,B)$ consistent with data as \eqref{set C: with A,Q,Zc}. 
To derive the main result from Petersen's lemma, a final reformulation of $\mathcal{C}$ is needed. 
We define
\begin{equation}
\label{set E}
\mathcal{E}:=\big\{\Zc+\mb{A}^{-1/2} \Upsilon \mb{Q}^{1/2} \colon \|\Upsilon\| \le 1 \big\},
\end{equation}
and show that it coincides with $\mathcal{C}$ in the next proposition.
\begin{proposition}
\label{proposition:set C cal = set E cal}
For $\mb{A} \succ 0$ and $\mb{Q} \succeq 0$, $\mathcal{C}=\mathcal{E}$.
\end{proposition}
\begin{proof}
It is sufficient to prove $\mathcal{E} \subseteq \mathcal{C}$ and $\mathcal{C} \subseteq \mathcal{E}$. \newline
$(\mathcal{E} \subseteq \mathcal{C})$ 
Suppose $Z \in \mathcal{E}$, i.e.,  $Z=\Zc+\mb{A}^{-1/2} \Upsilon \mb{Q}^{1/2}$ for some matrix $\Upsilon$ with $\|\Upsilon\| \leq 1$. Hence, $(Z-\Zc)^\top \mb{A}(Z-\Zc)=(\mb{A}^{-1/2} \Upsilon \mb{Q}^{1/2})^\top \mb{A}(\mb{A}^{-1/2} \Upsilon \mb{Q}^{1/2})=\mb{Q}^{1/2} \Upsilon^\top \Upsilon \mb{Q}^{1/2} \preceq \mb{Q}$.
Thus $Z \in \mathcal{C}$.\newline
$(\mathcal{C} \subseteq \mathcal{E})$ Suppose $Z \in \mathcal{C}$, i.e.,
\begin{equation}
\label{Z in set C}
(Z - \Zc)^\top \mb{A} (Z - \Zc) \preceq \mb{Q}.
\end{equation}
We need to find a matrix $\Upsilon$ with $\|\Upsilon\| \leq 1$ such that $Z=\Zc+\mb{A}^{-1/2} \Upsilon \mb{Q}^{1/2}$, i.e.,
\begin{equation}
\label{Upsilon in set E:equat}
\Upsilon \mb{Q}^{1/2}=\mb{A}^{1/2}(Z-\Zc).
\end{equation} 
If $\mb{Q}^{1/2} = 0$, we have the trivial solution $\Upsilon =0$. 
Otherwise, $\mb{Q}^{1/2}$ has $p \in \{1,\dots, n\}$ positive eigenvalues that define $\Lambda_p := \diag(\lambda_1,\dots,\lambda_p) \succ 0$.
Since $\mb{Q}^{1/2}$ is symmetric, there exists a real orthogonal matrix $T$ (i.e., $T^\top T = T T^\top = I$) such that 
\begin{equation}
\label{eigendec of sqrt Q}
\mb{Q}^{1/2} = T \Lambda T^\top :=  T \smat{\Lambda_p & 0\\ 0 & 0} T^\top,
\end{equation}
which is an eigendecomposition of $\mb{Q}^{1/2}$ and admits $\Lambda = \Lambda_p$ if $p=n$ (i.e., $\mb{Q}^{1/2} \succ 0$). 
Writing $T=: \big[ T_1~ T_2 \big]$ yields
\begin{equation}
\label{real orthog of T}
\smat{ T_1^\top T_1 & T_1^\top T_2\\ T_2^\top T_1 & T_2^\top T_2}
=
\smat{I & 0\\ 0 & I} \text{ and } T_1 T_1^\top + T_2 T_2^\top = I
\end{equation}
from $T^\top T = I$ and $T T^\top = I$. 
Select
\begin{equation}
\label{selection Upsilon}
\Upsilon = \mb{A}^{1/2} (Z-\Zc) T_1 \Lambda_p^{-1} T_1^\top
\end{equation}
(which reduces to $\mb{A}^{1/2} (Z-\Zc) \mb{Q}^{-1/2}$ if $p=n$).
We first show $\| \Upsilon \| \le 1$:
\begin{align*}
& \Upsilon^\top \Upsilon = T_1 \Lambda_p^{-1} T_1^\top (Z-\Zc)^\top  \mb{A}^{1/2}  \mb{A}^{1/2} (Z-\Zc) T_1 \Lambda_p^{-1} T_1^\top\\
& \overset{\eqref{Z in set C} }{\preceq}  T_1 \Lambda_p^{-1} T_1^\top \cdoT \mb{Q} [\star]^\top \overset{\eqref{eigendec of sqrt Q}}{=} T_1 \Lambda_p^{-1} T_1^\top \smat{T_1 & T_2} \cdoT \smat{\Lambda_p^2 & 0\\ 0 & 0} [\star]^\top\\
& = T_1 \Lambda_p^{-1} T_1^\top T_1 \cdoT \Lambda_p^2 [\star]^\top
\overset{\eqref{real orthog of T}}{=} T_1 T_1^\top \overset{\eqref{real orthog of T}}{\preceq} I.
\end{align*}
Then, we show that \eqref{Upsilon in set E:equat} holds. \eqref{Upsilon in set E:equat} is equivalent to
\begingroup%
\thinmuskip=.8mu plus 4mu
\medmuskip=1.6mu plus 8mu
\thickmuskip=2.4mu plus 12mu
\begin{align*}
& \Upsilon \mb{Q}^{1/2} \overset{\eqref{eigendec of sqrt Q}}{=} \Upsilon \smat{T_1 & T_2} \smat{\Lambda_p & 0\\ 0 & 0} T^\top  =\mb{A}^{1/2}(Z-\Zc)\\
& \iff \smat{\Upsilon T_1 \Lambda_p & 0}  =\mb{A}^{1/2}(Z-\Zc) T =\mb{A}^{1/2}(Z-\Zc) \smat{T_1 & T_2}\\
& \iff  \big(  \Upsilon T_1 \Lambda_p =\mb{A}^{1/2}(Z-\Zc)T_1
\text{, }  0 =\mb{A}^{1/2}(Z-\Zc) T_2 \big).
\end{align*}
\endgroup
If we show the last two equalities, we have shown \eqref{Upsilon in set E:equat} and completed the proof.
The first equality holds by the selection of $\Upsilon$ since $\Upsilon T_1 \Lambda_p  \overset{\eqref{selection Upsilon}}{=}
\mb{A}^{1/2} (Z-\Zc) T_1 \Lambda_p^{-1} T_1^\top T_1 \Lambda_p \overset{\eqref{real orthog of T}}{=} \mb{A}^{1/2} (Z-\Zc) T_1$. 
The second equality holds since the columns of $T_2$ are in $\ker \mb{Q}^{1/2}$ and $\ker \mb{Q}^{1/2}  \subseteq \ker  (\mb{A}^{1/2} (Z-\Zc))$. The columns of $T_2$ are in $\ker \mb{Q}^{1/2}$ because
$
\mb{Q}^{1/2} T_2 \overset{\eqref{eigendec of sqrt Q}}{=}  T \smat{\Lambda_p & 0\\ 0 & 0} \smat{T_1^\top \\ T_2^\top} T_2 \overset{\eqref{real orthog of T}}{=} T \smat{\Lambda_p & 0\\ 0 & 0} \smat{0 \\ I} = 0$;
$\ker \mb{Q}^{1/2}  \subseteq \ker  (\mb{A}^{1/2} (Z-\Zc))$ because, if $v$ satisfies $\mb{Q}^{1/2} v = 0$, then
$
0= v^\top\!\mb{Q}v \overset{\eqref{Z in set C}}{\ge} v^\top (Z - \Zc)^\top \mb{A} (Z - \Zc)v = |\mb{A}^{1/2} (Z - \Zc)v|^2
$, hence $\mb{A}^{1/2} (Z - \Zc)v = 0$.
\end{proof}

Considering $\mb{Q} \succeq 0$ rather than $\mb{Q} \succ 0$ is motivated since it allows us to include seamlessly the relevant special case of ideal data, namely, when the disturbance is not present. This corresponds indeed to $\Delta = 0$ and $\mathcal{D} = \{ 0 \}$ in~\eqref{set D} and $\mb{Q} = 0$ in~\eqref{Q simplest} by $D_0 \in \mathcal{D}$.
With the equivalent parametrization $\mathcal{E}$ of set $\mathcal{C}$ and Petersen's lemma in Fact~\ref{fact:petersen-ext}, we reach the next main result.
\begin{theorem}
\label{thm:sol}
For data given by $U_0$, $X_0$, $X_1$ in~\eqref{data} satisfying Assumption~\ref{ass:pers exc} and yielding $\mb{A}$, $\mb{B}$, $\mb{C}$ in~\eqref{A B C consist ellips}, feasibility of \eqref{probl} is equivalent to feasibility of
\begin{subequations}
\label{sol}
\begin{align}
& \text{find} & & Y, P=P^\top \succ 0 \label{sol:find}\\
& \text{s.~t.} & &  
\bmat{-P-\mb{C} & 0 & \mb{B}^\top\\
0 & -P & \bmat{P & Y^\top} \\
\mb{B} & \bmat{P\\ Y} & -\mb{A}
} \prec 0. \label{sol:lmi}
\end{align}
\end{subequations}
If \eqref{sol} is solvable, the controller gain is $K=YP^{-1}$.
\end{theorem}
\begin{proof}
Thanks to Proposition~\ref{proposition:set C cal = set E cal}, \eqref{probl:for all} is equivalent to the fact that  for all $(A,B) \in \mathcal E$
\begin{align*}
& (A+BK) P (A+BK)^\top - P \\
& = \bmat{A & B} \bmat{I\\K} P P^{-1} P \bmat{I\\K}^\top  \bmat{A & B}^\top - P \prec 0.
\end{align*}
Finding $P=P^\top\succ 0$, $K$ so that this matrix inequality holds for all $(A,B) \in \mathcal{E}$ is equivalent to finding $P=P^\top \succ 0$, $Y$ so that
\begin{equation}%
\label{aux lmi}
\bmat{- P & -\smat{A & B} \smat{P\\ Y}\\
-\smat{P\\Y}^\top  \smat{A & B}^\top & -P } \prec 0  \text{ for all } (A,B)\in \mathcal{E},
\end{equation}
by $P\succ 0$ and Schur complement. 
Note that, as claimed in the statement, $Y$ and $K$ are related by $Y= KP$, and $Y$ is preferred over $K$ as decision variable since $KP$ makes the matrix inequality nonlinear. 
$\big[ A~~B\big] = Z^\top \in \mathcal{E}$ if and only if $Z=\Zc+\mb{A}^{-1/2} \Upsilon \mb{Q}^{1/2}$ for some $\Upsilon$ with $\Upsilon^\top \Upsilon \preceq I$, by the parametrization in~\eqref{set E}. 
Hence, \eqref{aux lmi} is true if and only if \eqref{before petersen}, which is displayed below over two columns, holds for all $\Upsilon$ with $\Upsilon^\top \Upsilon \preceq I$.
\begin{figure*}[h]
\begingroup%
\thinmuskip=0mu plus 1mu
\medmuskip=0mu plus 2mu
\thickmuskip=1mu plus 3mu
\begin{equation}
\label{before petersen}
0 \succ \bmat{- P & -(\Zc+\mb{A}^{-1/2} \Upsilon \mb{Q}^{1/2})^\top \smat{P\\Y}\\
\star & -P } = 
\bmat{- P & -\Zc^\top \smat{P\\Y}\\
\star & -P }  + \bmat{0\\ - \smat{P\\Y}^\top \mb{A}^{-1/2}} \Upsilon \bmat{\mb{Q}^{1/2} & 0}  + \bmat{\mb{Q}^{1/2} \\ 0} \Upsilon^\top\hspace*{-.6mm} \bmat{0 & -\mb{A}^{-1/2} \smat{P\\Y} } 
\end{equation}
\hrule
\endgroup%
\end{figure*}
\eqref{before petersen} is written in a way that enables applying Petersen's lemma in Fact~\ref{fact:petersen-ext} with respect to the uncertainty $\Upsilon$. 
Indeed, simple computations yield that \eqref{before petersen} holds for all $\Upsilon$ with $\Upsilon^\top \Upsilon \preceq I$ if and only if there exists $\lambda > 0$ such that
\begin{equation}
\label{after petersen}
\bmat{- P + \lambda^{-1} \mb{Q} & -\Zc^\top \smat{P\\Y}\\
-\smat{P\\Y}^\top \Zc  & -P + \lambda \smat{P\\Y}^\top \mb{A}^{-1}   \smat{P\\Y} }  \prec 0.
\end{equation}
In summary, we have so far that \eqref{probl} is the same as
\begin{equation}
\label{proof feas prob aux 1}
\text{find } Y, P=P^\top \succ 0, \lambda>0 \text{ subject to } \eqref{after petersen}.
\end{equation}
Multiply both sides of \eqref{after petersen} by $\lambda>0$ and ``absorb'' it in $P$ and $Y$, so that \eqref{proof feas prob aux 1} is actually equivalent to
\begin{subequations}
\label{proof feas prob aux 2}
\begin{align}
& \text{find} & & Y, P=P^\top\succ 0 \\
& \text{s.~t.} & &  
\bmat{- P + \mb{Q} & -\Zc^\top \smat{P\\Y}\\
-\smat{P\\Y}^\top \Zc  & -P + \smat{P\\Y}^\top \mb{A}^{-1}   \smat{P\\Y} }  \prec 0. \label{proof feas prob aux 2:matr ineq}
\end{align}
\end{subequations}
Substitute in \eqref{proof feas prob aux 2:matr ineq} $\Zc$ and $\mb{Q}$ as in~\eqref{Zc,Q} to obtain
\begin{align*}
& \bmat{- P + \mb{B}^\top \mb{A}^{-1}\mb{B}-\mb{C} &  \mb{B}^\top \mb{A}^{-1} \smat{P\\Y}\\
\smat{P\\Y}^\top \mb{A}^{-1}\mb{B} & -P + \smat{P\\Y}^\top \mb{A}^{-1}   \smat{P\\Y} }  \\
& = \bmat{- P -\mb{C} &  0\\
0 & -P } + 
\bmat{ \mb{B}^\top \\ \smat{P\\Y}^\top } 
\mb{A}^{-1}
\bmat{ \mb{B} &  \smat{P\\Y}} 
\prec 0.
\end{align*}
Take a Schur complement of this inequality and replace by it the one in~\eqref{proof feas prob aux 2:matr ineq} to make \eqref{proof feas prob aux 2} equivalent to \eqref{sol}.
\end{proof}

Similarly, we use the set $\mathcal{E}$ in~\eqref{set E} and Petersen's lemma reported in Fact~\ref{fact:petersen-ext} to resolve \eqref{probl-ct} in the next theorem.
\begin{theorem}
\label{thm:sol-ct}
For data given by $U_0$, $X_0$, $X_1$ in~\eqref{data} satisfying Assumption~\ref{ass:pers exc} and yielding $\mb{A}$, $\mb{B}$, $\mb{C}$ in~\eqref{A B C consist ellips}, feasibility of \eqref{probl-ct} is equivalent to feasibility of
\begin{subequations}
\label{sol-ct}
\begin{align}
& \text{find} & & Y, P=P^\top \succ 0 \label{sol-ct:find} \\
& \text{s.~t.} & &  
\bmat{-\mb{C} & \mb{B}^\top - \bmat{P\\ Y}^\top \\
\mb{B} - \bmat{P\\ Y}  & -\mb{A}
} \prec 0. \label{sol-ct:lmi}
\end{align}
\end{subequations}
If \eqref{sol-ct} is solvable, the controller gain is $K=YP^{-1}$.
\end{theorem}
\begin{proof}
\submission{%
The proof follows the same reasoning of the proof of Theorem~\ref{thm:sol} and has somehow simplified steps since we do not need to first apply a Schur complement. It is thus omitted, but can be found in \cite{arXivVersion}.
}
\arxiv{%
Thanks to Proposition~\ref{proposition:set C cal = set E cal}, \eqref{probl-ct:for all} is equivalent to the fact that  for all $(A,B) \in \mathcal E$
\begin{align*}
& (A+BK) P + P (A+BK)^\top  \\
& = \bmat{A & B} \bmat{I\\K} P + P \bmat{I\\K}^\top \bmat{A & B}^\top \prec 0.
\end{align*}
Finding $K$, $P=P^\top\succ 0$ so that this matrix inequality holds for all $(A,B) \in \mathcal{E}$ is equivalent to finding $Y$, $P=P^\top \succ 0$ so that
\begin{equation}%
\label{aux lmi-ct}
\bmat{A & B} \bmat{P\\ Y} + \bmat{P\\Y}^\top \bmat{A & B}^\top \prec 0  \text{ for all } (A,B)\in \mathcal{E}.
\end{equation}
$\big[ A~~B\big] = Z^\top \in \mathcal{E}$ if and only if $Z=\Zc+\mb{A}^{-1/2} \Upsilon \mb{Q}^{1/2}$ for some $\Upsilon$ with $\Upsilon^\top \Upsilon \preceq I$, by the parametrization in~\eqref{set E}. 
Hence, \eqref{aux lmi-ct} is true if and only if \eqref{before petersen-ct}, which is displayed below over two columns, holds for all $\Upsilon$ with $\Upsilon^\top \Upsilon \preceq I$.
\begin{figure*}[h]
\begingroup%
\thinmuskip=0mu plus 1mu
\medmuskip=0mu plus 2mu
\thickmuskip=1mu plus 3mu
\begin{equation}
\label{before petersen-ct}
(\Zc+\mb{A}^{-1/2} \Upsilon \mb{Q}^{1/2})^\top \smat{P\\ Y} + \smat{P\\Y}^\top (\Zc+\mb{A}^{-1/2} \Upsilon \mb{Q}^{1/2})
= 
\Zc^\top \smat{P\\ Y} + \smat{P\\Y}^\top \Zc
+\smat{P\\Y}^\top  \mb{A}^{-1/2} \Upsilon \mb{Q}^{1/2} 
+ \mb{Q}^{1/2} \Upsilon^\top\hspace*{-.6mm} \mb{A}^{-1/2} \smat{P\\Y} 
\prec 0
\end{equation}
\endgroup%
\hrule
\end{figure*}
As in the proof of Theorem~\ref{thm:sol}, we apply to~\eqref{before petersen-ct} Petersen's lemma in Fact~\ref{fact:petersen-ext} with respect to the uncertainty $\Upsilon$; by simple computations, \eqref{before petersen-ct} holds for all $\Upsilon$ with $\Upsilon^\top \Upsilon \preceq I$ if and only if there exists $\lambda > 0$ such that
\begin{equation}
\label{after petersen-ct}
\Zc^\top \smat{P\\Y} + \smat{P\\Y}^\top \Zc + \lambda \smat{P\\Y}^\top \mb{A}^{-1} \smat{P\\Y} + \lambda^{-1} \mb{Q} \prec 0.
\end{equation}
In summary, \eqref{probl-ct} is the same as
\begin{equation}
\label{proof feas prob aux 1-ct}
\text{find } Y, P=P^\top\succ 0, \lambda > 0 \text{ subject to } \eqref{after petersen-ct}.
\end{equation}
Multiply both sides of \eqref{after petersen-ct} by $\lambda>0$ and ``absorb'' it in $P$ and $Y$, so that \eqref{proof feas prob aux 1-ct} is actually equivalent to
\begin{subequations}
\label{proof feas prob aux 2-ct}
\begin{align}
& \text{find} & & Y, P=P^\top\succ 0 \\
& \text{s.~t.} & &  
\Zc^\top \smat{P\\Y} + \smat{P\\Y}^\top \Zc + \smat{P\\Y}^\top \mb{A}^{-1} \smat{P\\Y} + \mb{Q} \prec 0. \label{proof feas prob aux 2-ct:matr ineq}
\end{align}
\end{subequations}
Substitute in \eqref{proof feas prob aux 2-ct:matr ineq} the expression of $\Zc$ and $\mb{Q}$ in~\eqref{Zc,Q} to obtain
\begin{align*}
& - \mb{B}^\top \mb{A}^{-1} \smat{P\\Y} - \smat{P\\Y}^\top \mb{A}^{-1} \mb{B} + \smat{P\\Y}^\top \mb{A}^{-1} \smat{P\\Y} -\mb{C}\\
& \hspace*{3mm} + \mb{B}^\top \mb{A}^{-1}\mb{B}   = -\mb{C} + (\mb{B} - \smat{P\\Y})^\top \mb{A}^{-1} (\mb{B} - \smat{P\\Y}) \prec 0.
\end{align*}
Taking a Schur complement of this matrix inequality and replacing by it the one in~\eqref{proof feas prob aux 2-ct:matr ineq} makes \eqref{proof feas prob aux 2-ct} equivalent to \eqref{sol-ct}.}
\end{proof}

Suppose that the set $\mathcal{C}$ is given directly in the form~\eqref{set C: with A,Q,Zc} as a matrix-ellipsoid over-approximation of a less tractable set that is derived from data, which we discuss in Section~\ref{sec:C as overapprox}.
For this case, a better alternative to Theorems~\ref{thm:sol}-\ref{thm:sol-ct} is the next corollary.
\begin{corollary}
\label{cor:assumeOnlyAandQ}
For the set $\mathcal{C} = \big\{ \big[A~~B\big] = Z^\top \colon 
(Z - \Zc)^\top \mb{A} (Z - \Zc) \preceq \mb{Q}
\big\}$ as in~\eqref{set C: with A,Q,Zc}, assume $\mb{A} \succ 0$ and $\mb{Q} \succeq 0$. Then, feasibility of \eqref{probl} (resp., \eqref{probl-ct}) is equivalent to feasibility of \eqref{sol} (resp., \eqref{sol-ct}).
If \eqref{sol} (resp., \eqref{sol-ct}) is solvable, the controller gain is $K=YP^{-1}$.
\end{corollary}
\begin{proof}
Proposition~\ref{proposition:set C cal = set E cal} is true under $\mb{A} \succ 0$ and $\mb{Q} \succeq 0$, and only Propositon~\ref{proposition:set C cal = set E cal} is used in Theorems~\ref{thm:sol}-\ref{thm:sol-ct}.
\end{proof}

\arxiv{
\begin{remark}
The conditions in Theorems~\ref{thm:sol}-\ref{thm:sol-ct} depend on the parametrization of $\mathcal{C}$ through $\mb{A}$, $\mb{B}$, $\mb{C}$ in~\eqref{set C:with A,B,C}. 
Alternatively, we can give conditions depending on the parametrization of $\mathcal{C}$ through $\mb{A}$, $\Zc$, $\mb{Q}$ in~\eqref{set C: with A,Q,Zc} as
\begin{align}
&  \smat{-P + \mb{Q} & -\Zc^\top \smat{P\\ Y} & 0\\
- \smat{P \\ Y}^\top \Zc & -P & \smat{P \\ Y}^\top \\
0 & \smat{P\\ Y} & -\mb{A}
} \prec 0
\label{sol-dt:decrease alt}\\
& \smat{\smat{P \\ Y}^\top  \Zc + \Zc^\top \smat{P \\ Y} + \mb{Q} & \hspace*{3mm}\smat{P \\ Y}^\top \\
\smat{P \\ Y} & -\mb{A}
} \prec 0.
\label{sol-ct:decrease alt}
\end{align}
\eqref{sol-dt:decrease alt} and \eqref{sol-ct:decrease alt} are equivalent forms of \eqref{sol:lmi} and \eqref{sol-ct:lmi}, respectively, and can be obtained from \eqref{proof feas prob aux 2:matr ineq} and \eqref{proof feas prob aux 2-ct:matr ineq} by taking Schur complements.
These two conditions can be more convenient numerically.
\end{remark}
}
\submission{
\begin{remark}
\label{remark:alternative conditions}
Instead of parameters $\mb{A}$, $\mb{B}$, $\mb{C}$ of $\mathcal{C}$ in~\eqref{set C:with A,B,C}, we can write the conditions \eqref{sol:lmi} and \eqref{sol-ct:lmi} in Theorems~\ref{thm:sol} and \ref{thm:sol-ct} in terms of $\mb{A}$, $\Zc$, $\mb{Q}$ of $\mathcal{C}$ in~\eqref{set C: with A,Q,Zc} as
\begingroup%
\setlength\arraycolsep{2.5pt}%
\thinmuskip=0.5mu plus 1mu
\medmuskip=1mu plus 2mu
\thickmuskip=1.5mu plus 3mu
\begin{align*}
&  \smat{-P + \mb{Q} & \star & \star\\
- \smat{P \\ Y}^\top \Zc & -P & \star \\
0 & \smat{P\\ Y} & -\mb{A}
} \prec 0 
\text{ and } \smat{\smat{P \\ Y}^\top  \Zc + \Zc^\top \smat{P \\ Y} + \mb{Q} & \star\\
\smat{P \\ Y} & -\mb{A}
} \prec 0.
\end{align*}
\endgroup
These conditions, which are obtained by Schur complement (see~\eqref{proof feas prob aux 2:matr ineq}), are equivalent to \eqref{sol:lmi} and \eqref{sol-ct:lmi}, respectively, and can be more convenient numerically.
\end{remark}
}

\section{Discussion and interpretations}
\label{sec:lin discussion}

This section is devoted to giving an overall interpretation of the previous developments.

\subsection{Assumption~\ref{ass:pers exc} and persistence of excitation}
\label{sec:pers of exc}

Assumption~\ref{ass:pers exc} is intimately related to the notion of
persistence of excitation, as we now motivate. 
With full details in \cite[\S 4.2]{dpt2021Aut}, the result \cite[Cor.~2]{willems2005note}, which was given in the ideal case without disturbance $x^+ = A_\star x + B_\star u$, can show for the present case
\begin{equation*}
x^+ = A_\star x + B_\star u + d = A_\star x + \big[B_\star~~I\big]
\smat{u \vspace*{1.5mm}\\ d}
\end{equation*}
that: \textit{(i)}~controllability of 
$(A_\star,B_\star)$, \textit{(ii)}~an input sequence 
persistently exciting of order $n+1$ \footnote{See \cite[p.~327]{willems2005note} 
or \cite[Def.~1]{depersis2020tac} for a definition.}, and \textit{(iii)}~a disturbance sequence 
persistently exciting of order $n+1$ imply together that 
$\smat{X_0\\ U_0\\ D_0}$ has full row rank and so has $\smat{X_0\\ U_0}$, as required in Assumption~\ref{ass:pers exc}. 
In the ideal case, \textit{(i)} and \textit{(ii)} imply that $\smat{X_0\\ U_0}$ has full row rank \cite[Cor.~2]{willems2005note}.
In other words, Assumption~\ref{ass:pers exc} holds under a persistently exciting disturbance and the same conditions of the ideal case, which include a persistently exciting input.

\subsection{Ellipsoidal uncertainty, least squares and certainty-equivalence control}
\label{sec:least squares}

\arxiv{From \eqref{sol-dt:decrease alt} and \eqref{sol-ct:decrease alt}, t}\submission{T}he discrete- and continuous-time 
stability conditions of Theorems~\ref{thm:sol} and \ref{thm:sol-ct} are
equivalent\submission{, see Remark~\ref{remark:alternative conditions},} to
\begingroup%
\setlength\arraycolsep{2.5pt}%
\thinmuskip=0.3mu plus 1mu
\medmuskip=.6mu plus 2mu
\thickmuskip=.9mu plus 3mu
\begin{align}
& \hspace*{-8pt} \bmat{-P  & -\Zc^\top \smat{P\\ Y} \\
- \smat{P \\ Y}^\top \Zc & -P
} 
+
\bmat{\mb{Q} & 0 \\
0 & \smat{P \\ Y}^\top \mb{A}^{-1} \smat{P\\ Y}
} 
\prec 0~\text{ and } \hspace*{-2pt} 
\label{sol-dt:decrease alt LS}\\
& \hspace*{-8pt} 
\Bigg(\bmat{P \\ Y}^\top  \Zc + \Zc^\top \bmat{P \\ Y}\Bigg) + \Bigg( \mb{Q}
+ \bmat{P \\ Y}^\top \mb{A}^{-1} \bmat{P \\ Y} \Bigg) \prec 0, \hspace*{-2pt}  \label{sol-ct:decrease alt LS}
\end{align}
\endgroup
respectively, with $\mb{A}$ as in \eqref{A B C consist ellips} and $\Zc$, $\mb{Q}$ as in \eqref{Zc,Q}.
The matrix $\Zc$ appears only in the first term of the two matrix inequalities and it represents the center of the uncertainty set $\mathcal C$, see \eqref{set C: with A,Q,Zc}. 
On the other hand, the matrices $\mb{A}$, $\mb{Q}$ appearing in the second term of the two matrix inequalities determine the size of the uncertainty; in particular, the size of 
$\mathcal{C}$ is given by
$(\det \mb{Q} )^{(n+m)/2} (\det \mb{A} )^{-n/2}$,
see \cite[\S 2.2]{Bisoffi2021tradeoffs}. 
By Lemma~\ref{lemma:sign of A,Q}, the second terms in~\eqref{sol-dt:decrease alt LS} and \eqref{sol-ct:decrease alt LS} are positive semidefinite, and this means that the design problem can be interpreted as the problem of finding a controller that \emph{robustly} stabilizes the dynamics associated with the center $\Zc$ of the uncertainty set $\mathcal C$, where the uncertainty increases with the noise bound $\Delta$, see the expression of $\mb{Q}$ in \eqref{Q simplest}.

Quite interestingly, the center $\Zc$ of the uncertainty set $\mathcal C$ coincides with the (ordinary) least-squares estimate of the system dynamics, i.e., with the solution $(A_{\tu{ls}}, B_{\tu{ls}})$ to
\begin{equation*}
\min_{(A,B)} \|  X_1 - A X_0 - B U_0  \|_{\tu{F}}^2
\end{equation*}
where $\| \cdot  \|_{\tu{F}}$ denotes the Frobenius norm. Indeed,
\begin{align*}
\big[ A_{\tu{ls}}~~B_{\tu{ls}} \big] & := \arg \min_{(A,B)} \|  X_1 - A X_0 - B U_0  \|_{\tu{F}}^2 \\ 
& = X_1 \smat{X_0\\ U_0}^\dagger  = ( - \mb{A}^{-1} \mb{B} )^\top = \Zc^\top,
\end{align*}
see \cite[\S 2.6]{verhaegen2007filtering}. 
This justifies why \emph{certainty-equivalence} control works well in regimes of small uncertainty (when $\Delta$ is small), in agreement with what has been recently observed in \cite{mania2019certainty,dorfler2021certainty}. 
On the other hand, this also explains why with noisy data robust control is generally needed, which is also the main idea behind the \emph{robust indirect} control approaches
\cite{Dean2018journ,ferizbegovic2019learning,treven2021learning} under a stochastic noise description. 
Besides the noise description, a difference between our work and \cite{Dean2018journ,ferizbegovic2019learning,treven2021learning} is that our approach is direct in the sense that solving \eqref{sol} or \eqref{sol-ct} does not require to explicitly constructing any estimate of the system dynamics, which is distinctive of indirect methods.

\subsection{Comparison with alternative conditions in \cite{vanwaarde2020noisy}}
\label{sec:comparison}

Sections~\ref{sec:pers of exc} leads us 
to a comparison with the approach based on a matrix S-procedure in \cite{vanwaarde2020noisy}.
We recall its main result for data-based stabilization, \cite[Thm.~14]{vanwaarde2020noisy}, and rephrase it for the context of this paper in the next fact.

\begin{fact}{\cite[Thm.~14]{vanwaarde2020noisy}}
\label{fact:data driven by S-lemma}
Assume that the generalized Slater condition
\begin{equation*}
\bmat{I\\ \bar Z}^\top 
\bmat{\mb{C} & \mb{B}^\top \\ \mb{B} & \mb{A}}
\bmat{I\\ \bar Z} \prec 0
\end{equation*}
holds for some $\bar Z \in \real^{(n+m) \times n}$. 
Then, there exist a feedback gain $K$ and a matrix $P=P^\top \succ 0$ such that $(A+ B K) P (A + B K)^\top - P \prec 0$ for all $(A,B) \in \mathcal{C}$ if and only if the next program is feasible
\begin{align*}
& \text{find} & & P = P^\top \succ 0, Y,  \alpha \ge 0, \beta >0\\
& \text{s.~t.} & &  
\bmat{-P + \beta I & 0\\ 0 & \smat{P\\ Y} P^{-1} \smat{P\\ Y}^\top }
- \alpha 
\bmat{\mb{C} & \mb{B}^\top \\ \mb{B} & \mb{A}}
\preceq 0.
\end{align*}
If $P$ and $Y$ are a solution to it, then $K = Y P^{-1}$ is a stabilizing gain for all $(A,B)\in \mathcal{C}$.
\end{fact}

Fact~\ref{fact:data driven by S-lemma} and Theorem~\ref{thm:sol} are two alternative approaches since both propose a necessary and sufficient condition for quadratic 
stabilization; indeed, \eqref{sol:lmi} in Theorem~\ref{thm:sol} is equivalent, by Schur complement and changing sign to off-diagonal terms, to
\begin{equation*}
\bmat{-P & 0\\ 0 & \smat{P\\ Y} P^{-1} \smat{P\\ Y}^\top }
- \bmat{\mb{C} & \mb{B}^\top \\ \mb{B} & \mb{A}}
\prec 0.
\end{equation*}
There are some interesting differences, though.
Fact~\ref{fact:data driven by S-lemma} operates under a Slater condition, whereas Theorem~\ref{thm:sol} under Assumption~\ref{ass:pers exc}.
The Slater condition can capture the case of an unbounded set $\mathcal{C}$, which cannot occur with Assumption~\ref{ass:pers exc} (see Lemma~\ref{lemma:boundedness of set C}); 
by contrast, the Slater condition cannot capture the case of ideal data \cite[\S II.C]{vanWaarde2021finsler}, which requires different arguments \cite{vanWaarde2021finsler}.
We believe the approach through Petersen's lemma is appealing due to the conceptual insights it provides on the data-based control laws, which we have expounded in this Section~\ref{sec:lin discussion}, and to its ease of applicability beyond the linear systems of Section~\ref{sec:data-driv lin:main results}, as we show for polynomial systems in Section~\ref{sec:data-driven polynomial}.

\subsection{$\mathcal{C}$ as an ellipsoidal over-approximation}
\label{sec:C as overapprox}

As \eqref{set C} shows, we have derived set $\mathcal{C}$ based on the disturbance bound in $\mathcal{D}$ and the relation data need to satisfy. 
On the other hand, the matrix-ellipsoid  form \eqref{set C: with A,Q,Zc} of set $\mathcal{C}$ can be fruitfully used as an over-approximation of sets of matrices consistent with data that are not matrix ellipsoids, since ellipsoidal sets are generally better tractable.
In that case, as long as matrices $\mb{A}$ and $\mb{Q}$ in~\eqref{set C: with A,Q,Zc} satisfy $\mb{A} \succ 0$ and $\mb{Q} \succeq 0$, one can use directly Corollary~\ref{cor:assumeOnlyAandQ}.
We describe succinctly a relevant case when this could be done based on \cite{Bisoffi2021tradeoffs}, to which we refer the reader for a more elaborate discussion.

With the definitions for $i = 0, 1, \dots , T-1$
\begin{align*}
\varkappa^\circ_i := x(t_{i+1}) \text{ or } \varkappa^\circ_i := \dot{x}(t_i), \varkappa_i:=x(t_i), \upsilon_i := u(t_i)
\end{align*}
that embed discrete or continuous time, consider the disturbance model
$\mathcal{D}_{\tu{i}} := \{ d \in \real^n \colon |d|^2 \le \delta \}$.
The corresponding set of matrices consistent with all data points $i=0,\dots, T-1$ is $\mathcal{I} := \bigcap_{i =0}^{T-1} \{ (A,B) \colon \varkappa^\circ_i = A \varkappa_i + B \upsilon_i +d , d\in \mathcal{D}_{\tu{i}} \}$ and, due to the intersection, its size remains equal or decreases with $T$.
$\mathcal{I}$ is not a matrix ellipsoid and the results in Section~\ref{sec:data-driv lin:main results} cannot be applied to it.
Still, a matrix ellipsoid $\mathcal{C} \supseteq \mathcal{I}$ as in~\eqref{set C:with A,B,C} can be readily obtained; its parameters $\mb{A}$, $\mb{B}$, $\mb{C} := \mb{B}^\top \mb{A}^{-1} \mb{B} - I$ follow from the optimization problem
\begin{subequations}
\label{opt overapprox ellips}
\begin{align}
& \text{min.} & & -\log\det \mb{A} \text{\quad(over $\mb{A}$, $\mb{B}$, $\tau_1$, \dots, $\tau_{T-1}$)}\\
& \text{s.~t.}  & & 
\hspace*{-5pt}\bmat{
- I - \sum_{i =0}^{T-1} \tau_i \gamma_i   
& \star
& \star\\
\mb{B} - \sum_{i =0}^{T-1} \tau_i \beta_i 
& \mb{A} - \sum_{i =0}^{T-1} \tau_i \alpha_i 
& \star \\
\mb{B}
& 0 
& - \mb{A}
} \preceq 0\\
& & & \mb{A} \succ 0, \, \tau_i\ge 0 \text{ for } i = 0, 1, \dots , T-1
\end{align}%
\end{subequations}
with data-related quantities
\begin{equation}
\label{gamma_i,beta_i,alpha_i}
\hspace*{-6pt} \gamma_i :=-\delta I +\varkappa^\circ_i {\varkappa^\circ_i}^\top\hspace*{-2pt},\hspace*{1pt}
\beta_i := -\bmat{\varkappa_i\\\upsilon_i}{\varkappa^\circ_i}^\top\hspace*{-2pt},\hspace*{1pt}
\alpha_i := \bmat{\varkappa_i \\ \upsilon_i}\hspace*{-1pt}\bmat{\varkappa_i\\\upsilon_i}^\top\hspace*{-2pt}\hspace*{-3pt}
\end{equation}
for $i=0,\dots, T-1$.
(This optimization problem is the natural extension to matrix ellipsoids of the one in~\cite[\S 3.7.2]{boyd1994linear} for classical ellipsoids.)
A feasible solution to~\eqref{opt overapprox ellips} guarantees by construction  $\mb{A} \succ 0$ and $\mb{Q} = \mb{B}^\top \mb{A}^{-1}\mb{B} - \mb{C} = I \succeq 0$ (by the selection of $\mb{C}$); hence,
Corollary~\ref{cor:assumeOnlyAandQ} can be applied to this $\mathcal{C}$.
A very desirable feature of this $\mathcal{C}$, inherited from $\mathcal{D}_{\tu{i}}$, is that its size generally decreases with $T$, and this requires, in turn, a lesser degree of robustness in the design of the controller if one collects more data.
In summary, when an instantaneous disturbance model $\mathcal{D}_{\tu{i}}$ is given, the results of Section~\ref{sec:data-driv lin:main results} cannot be applied to the corresponding set $\mathcal{I}$  but can be to the set $\mathcal{C}$ obtained by~\eqref{opt overapprox ellips}.
The tightness of the over-approximation is 
problem-dependent, and it might be convenient to work directly with $\mathcal{I}$ at the expense of an increase in the computational complexity \cite{Bisoffi2021tradeoffs}.

\section{Data-driven control for polynomial systems}
\label{sec:data-driven polynomial}

We illustrate in this section that Petersen's lemma proves useful also for polynomial systems, if applied pointwise. 
As an important class of a nonlinear input-affine system, consider the polynomial system
\begin{equation}
\label{sys-pol}
\dot x = f_\star(x) + g_\star(x) u + d = A_\star Z(x) + B_\star W(x) u + d
\end{equation}
where $x \in \real^n$ is the state, $u \in \real^m$ is the input, $d \in \real^n$ is a disturbance; $x \mapsto Z(x) \in \real^N$ is a \emph{known} regressor vector of monomials of $x$ and $x \mapsto W(x) \in \real^{M \times m}$ is a \emph{known} regressor matrix of monomials of $x$; the rectangular matrices $A_\star \in \real^{n \times N}$ and $B_\star \in \real^{n \times M}$ with the coefficients of the regressors are \emph{unknown} to us.
The selection of the regressors $Z$ and $W$ is a key aspect for feasibility of the optimization-based control law, and we comment this in detail in Section~\ref{sec:num ex:pol}.
We will handle data-driven control conditions for \eqref{sys-pol} through a sum-of-squares relaxation; since sum-of-square tools are most commonly used for continuous-time systems, we consider directly the continuous-time case in \eqref{sys-pol}.

As in Section~\ref{sec:probl form}, we perform an experiment on the system by applying an input sequence $u(t_0)$, \dots, $u(t_{T-1})$ of $T$ samples and measure the state and state-derivative sequences $x(t_0)$, \dots, $x(t_{T-1})$ and $\dot x(t_0)$, \dots, $\dot x(t_{T-1})$.
The unknown disturbance sequence $d(t_0)$, \dots, $d(t_{T-1})$ affects the evolution of the system, leading to noisy data. We collect the data points in the matrices
\begin{subequations}
\label{data-pol}
\begin{align}
V_0 & :=\bmat{W(x(t_0))u(t_0) & \cdots & W(x(t_{T-1}))u(t_{T-1})} \\
Z_0 & :=\bmat{Z(x(t_0)) & & \cdots & Z(x(t_{T-1}))} \\
X_1 & :=\bmat{\dot x(t_0) & \cdots & \dot x(t_{T-1})}.
\end{align}
\end{subequations}
With the unknown disturbance sequence in $D_0 := \bmat{d(t_0) & & \cdots & d(t_{T-1})}$, data satisfy
\begin{equation*}
X_1 = A_\star Z_0 + B_\star V_0 + D_0.
\end{equation*}
As in Section~\ref{sec:probl form}, the set of matrices consistent with data $X_1$, $Z_0$, $V_0$ and disturbance model $\mathcal{D}$ in~\eqref{set D} is
\begin{equation*}
\tilde{\mathcal{C}} := \{ (A,B) \colon X_1 = A Z_0 + B V_0 + D, D \in \mathcal{D} \}.
\end{equation*}
We can then follow closely the rationale of Section~\ref{sec:reform set C}, and we briefly outline only the key steps.
The set $\tilde{\mathcal{C}}$ can be reformulated as
\begin{align*}
& \tilde{\mathcal{C}} = \Big\{
(A,B)\colon
\bmat{
I & A & B
}
\cdoT
\left[ \begin{array}{c|c}
\tilde{\mb{C}} & \tilde{\mb{B}}^\top \\ 
\hline
\tilde{\mb{B}} & \tilde{\mb{A}}
\end{array} \right]
[\star]^\top
\preceq 0
\Big\} \\
& \left[ \begin{array}{c|c}
\tilde{\mb{C}} & \tilde{\mb{B}}^\top \\ 
\hline
\tilde{\mb{B}} & \tilde{\mb{A}}
\end{array} \right]
:=
\left[ \begin{array}{c|c}
X_1X_1^\top-\Delta\Delta^\top & -X_1 \smat{Z_0\\V_0}^\top \\
\hline
\phantom{\rule{0.1pt}{13pt}} 
-\smat{Z_0\\V_0}X_1^\top & \smat{Z_0\\ V_0} \smat{Z_0\\ V_0}^\top
\end{array} \right].
\end{align*}
The next assumption is analogous to Assumption~\ref{ass:pers exc}.
\begin{assumption}
\label{ass:pers ext-pol}
Matrix $\smat{Z_0\\ V_0}$ has full row rank.
\end{assumption}
$\tilde{\mb{A}} \succ 0$ by Assumption~\ref{ass:pers ext-pol}, and $\tilde{\mathcal{C}}$ can be rewritten as
\begin{align}
& \tilde{\mathcal{C}} = \big\{ \big[A~~B\big] = Z^\top \colon 
(Z - \tilde{\Zc})^\top \tilde{\mb{A}} (Z - \tilde{\Zc}) \preceq \tilde{\mb{Q}}
\big\} \label{set tilde C}\\
& \tilde{\Zc} := - \tilde{\mb{A}}^{-1} \tilde{\mb{B}},\, \tilde{\mb{Q}}:= \tilde{\mb{B}}^\top \tilde{\mb{A}}^{-1}\tilde{\mb{B}} - \tilde{\mb{C}}. \notag
\end{align}
The logical steps of Lemma~\ref{lemma:sign of A,Q}, Lemma~\ref{lemma:boundedness of set C} and Proposition~\ref{proposition:set C cal = set E cal} can be repeated in the same way after replacing $\smat{X_0 \\ U_0}$ with $\smat{Z_0\\ V_0}$, so their results are summarized in the next lemma without proof.
\begin{lemma}
Under Assumption~\ref{ass:pers ext-pol}, we have: $\tilde{\mb{A}} \succ 0$, $\tilde{\mb{Q}} \succeq 0$, $\tilde{\mathcal{C}}$ is bounded with respect to any matrix norm, and 
\begin{equation}
\label{set tilde C alt}
\tilde{\mathcal{C}} = \big\{\tilde{Z}_{\tu{c}}+\tilde{\mb{A}}^{-1/2} \Upsilon \tilde{\mb{Q}}^{1/2} \colon \|\Upsilon\| \le 1 \big\}.
\end{equation}
\end{lemma}

As in Section~\ref{sec:data-driv lin:main results}, the matrix-ellipsoid parametrization in~\eqref{set tilde C alt} is key to apply Petersen's lemma, which allows us to obtain the next result for data-driven control of the polynomial system in~\eqref{sys-pol}.
\begin{proposition}
\label{prop:sol-pol}
Let Assumption~\ref{ass:pers ext-pol} hold.
Given positive definite\footnote{That is, zero at zero and positive elsewhere.} polynomials $\ell_1$, $\ell_2$ with $\ell_1$ radially unbounded\footnote{That is, $\ell_1(x) \to +\infty$ as $|x| \to +\infty$.}, suppose there exist polynomials $V$, $k$, $\lambda$ with $V(0)=0$ and $k(0)=0$ such that for each $x$
\begin{align}
& \hspace*{-2mm}V(x) - \ell_1(x) \ge 0 \label{sol-pol:pd and rad unbnd} \\
& \hspace*{-4mm}
\resizebox{.435\textwidth}{!}{
\text{$\bmat{
\ell_2(x) + \dVx \tilde{Z}_{\tu{c}}^\top \smat{Z(x) \\ W(x) k(x) } & \star &  \star\\
\tilde{\mb{A}}^{-1/2}\smat{Z(x) \\ W(x) k(x) } & -\lambda(x) I & \star\\
\lambda(x) \tilde{\mb{Q}}^{1/2}\dVx^\top & 0 & -4 \lambda(x) I\hspace*{-1pt}
} \preceq 0$\hspace*{-2mm}}
} \label{sol-pol:decrease}\\
& \hspace*{-2mm}\lambda(x) >0. \label{sol-pol:multipl}
\end{align}
Then, the origin of
\begin{equation*}
\dot x = A Z(x) + B W(x) k(x)  =: f_{A,B}(x)
\end{equation*}
is globally asymptotically stable for all $(A,B) \in \tilde{\mathcal{C}}$, and in particular for $(A_\star,B_\star) \in \tilde{\mathcal{C}}$, i.e., for the closed loop $\dot x =  f_{A_\star,B_\star}(x)$.
\end{proposition}

Let us comment the conditions and the conclusion of Proposition~\ref{prop:sol-pol}. 
Condition~\eqref{sol-pol:pd and rad unbnd} imposes positive definiteness and radial unboundedness of the Lyapunov function $V$; condition~\eqref{sol-pol:multipl} is the positivity of the multiplier used in Petersen's lemma; condition~\eqref{sol-pol:decrease} imposes decrease of the Lyapunov function for all $(A,B)\in \tilde{\mathcal{C}}$. 
In particular, suppose $\tilde{Z}_{\tu{c}}^\top = \big[ A_\star~~B_\star \big]$ in~\eqref{sol-pol:decrease}; then, the block (1,1) alone of the matrix in~\eqref{sol-pol:decrease} would express a model-based condition for global asymptotic stability of $\dot x = A_\star Z(x) + B_\star W(x)$.
The conclusion is global asymptotic stability of the closed loop $\dot x = f_{A,B}(x)$ for all $(A,B) \in \tilde{\mathcal{C}}$. Similarly to the linear case (see comment below \eqref{probl}-\eqref{probl-ct}), this is relevant for the closed loop with disturbance $\dot x = f_{A_\star,B_\star}(x)+d$ obtained from~\eqref{sys-pol} because global asymptotic stability guarantees input-to-state stability with ``small disturbances'' as shown in \cite[Thm.~2]{sontag1990further}, to which we refer for precise characterizations.

\begin{proofof}{Proposition~\ref{prop:sol-pol}}
Note first that since $Z(0)=0$ ($Z$ is a regressor of monomials of $x$) and $k(0)=0$, the origin is an equilibrium of $f_{A,B}$ for all $(A,B) \in \tilde{\mathcal{C}}$.
Then, the proof consists of showing that $V$ is a Lyapunov function for all systems $\dot x = f_{A,B}(x)$, $(A,B) \in \tilde{\mathcal{C}}$. 
Specifically, we show that \textit{(i)} $V$ is positive definite and radially unbounded, and \textit{(ii)} its derivative along solutions satisfies
\begin{equation}
\label{sol-pol:decr cond orig}
\begin{split}
\langle \nabla V(x), f_{A,B}(x)\rangle & = \tfrac{\partial V}{\partial x} (x) \smat{A & B} \smat{Z(x) \\ W(x) k(x) } \\
& \le -\ell_2(x)~~\forall x, \forall (A,B) \in \tilde{\mathcal{C}}.
\end{split}
\end{equation}
If the previous properties \textit{(i)}-\textit{(ii)} hold, classical Lyapunov theory \cite[Thm.~4.2]{khalil2002nonlinear} yields the conclusion of the theorem. 
Positive definiteness of $V$ follows from $V(0)=0$, \eqref{sol-pol:pd and rad unbnd} and $\ell_1$ positive definite; radial unboundedness of $V$ follows from \eqref{sol-pol:pd and rad unbnd} and $\ell_1$ radially unbounded.
We then address the derivative along solutions of $V$.
Set $\big[A~~B\big] = Z^\top \in \tilde{\mathcal{C}}$ in~\eqref{sol-pol:decr cond orig} and substitute the parametrization of $Z$ from~\eqref{set tilde C alt}; \eqref{sol-pol:decr cond orig} holds if and only if, for each $x$,
\begingroup%
\thinmuskip=0.5mu plus 1mu
\medmuskip=1.mu plus 2mu
\thickmuskip=1.5mu plus 3mu
\begin{align}
& - \ell_2(x) \ge \langle \nabla V(x), f_{A,B}(x)\rangle  = \dVxS \tilde{Z}_{\tu{c}}^\top \smat{Z(x) \\ W(x) k(x) } 
\notag \\
& + \smat{Z(x) \\ W(x) k(x) }^\top \tilde{\mb{A}}^{-1/2} \Upsilon \tilde{\mb{Q}}^{1/2} \tfrac{1}{2} \dVxS^\top \label{sol-pol:decr for all Upsilon} \\
& + \tfrac{1}{2} \dVxS \tilde{\mb{Q}}^{1/2} \Upsilon^\top \tilde{\mb{A}}^{-1/2} \smat{Z(x) \\ W(x) k(x) }  \hspace*{5pt}\forall \Upsilon \text{ with } \|\Upsilon\|\le 1. \notag
\end{align}
\endgroup%
We now show that this is true thanks to~\eqref{sol-pol:decrease} and \eqref{sol-pol:multipl}.
By Schur complement for nonstrict inequalities \cite[p.~28]{boyd1994linear} and \eqref{sol-pol:multipl}, \eqref{sol-pol:decrease} is equivalent to 
\begin{align}
& -\ell_2(x) \ge  \dVx \tilde{Z}_{\tu{c}}^\top \smat{Z(x) \\ W(x) k(x) } \notag \\
& +  \smat{Z(x) \\ W(x) k(x) }^\top  \cdoT \frac{\tilde{\mb{A}}^{-1}}{\lambda(x)}  [\star]^\top + \dVx \cdoT \frac{\lambda(x)\tilde{\mb{Q}}}{4} [\star]^\top.  \label{sol-pol:decr before petersen}
\end{align}
In other words, we have by~\eqref{sol-pol:decrease} and \eqref{sol-pol:multipl} that for each $x$, there exists $1/\lambda(x) > 0$ such that \eqref{sol-pol:decr before petersen} holds.
Apply Fact~\ref{fact:petersen-nonstrict} pointwise (i.e., for each $x$) to~\eqref{sol-pol:decr before petersen} with $\mb{E}$ and $\mb{G}^\top$ corresponding respectively to $\smat{Z(x) \\ W(x) k(x) }^\top \hspace*{-2pt} \tilde{\mb{A}}^{-1/2}$ and $\frac{1}{2} \dVx \tilde{\mb{Q}}^{1/2}$; the fact that for each $x$, there exists $1/\lambda(x) > 0$ such that \eqref{sol-pol:decr before petersen} holds implies that for each $x$, \eqref{sol-pol:decr for all Upsilon} holds or, equivalently, that for each $x$, \eqref{sol-pol:decr cond orig} holds.
All properties required of $V$ have been shown, and the conclusion of the proposition follows.
\end{proofof}

When writing the Lyapunov derivative along solutions as in~\eqref{sol-pol:decr cond orig} and substituting the expression of $\tilde{\mathcal{C}}$ as in~\eqref{sol-pol:decr for all Upsilon}, the utility of Petersen's lemma beyond the case of linear systems becomes clear.
We use the nonstrict version of it in Fact~\ref{fact:petersen-nonstrict} (instead of the strict version in Fact~\ref{fact:petersen-ext}) in view of the next sum-of-squares relaxation and the subsequent numerical implementation, where only \emph{nonstrict} inequalities can effectively be implemented.
Polynomial positivity in the conditions of Proposition~\ref{prop:sol-pol} is impractical to verify, so we turn them into sum-of-squares conditions in the next theorem.
\begin{theorem}
\label{thm:sos}
Let Assumption~\ref{ass:pers ext-pol} hold.
Given positive definite polynomials $\ell_1$, $\ell_2$ with $\ell_1$ radially unbounded and  a positive scalar $\epsilon_\lambda$, suppose there exist polynomials $V$, $k$, $\lambda$ with $V(0)=0$ and $k(0)=0$ such that
\begin{subequations}
\label{sos}
\begin{align}
& V - \ell_1 \in \mathcal{S} \label{sos:pd and rad unbnd} \\
&
 - \bmat{
\ell_2 + \frac{\partial V}{\partial x} \tilde{Z}_{\tu{c}}^\top \smat{Z \\ W k} & \star &  \star\\
\tilde{\mb{A}}^{-1/2} \smat{Z \\ W k} & -\lambda I & \star\\
\lambda \tilde{\mb{Q}}^{1/2} \frac{\partial V}{\partial x}^\top & 0 & -4 \lambda I
} \in\mathcal{S}_{\tu{m}}
\label{sos:decrease}\\
& \lambda -\epsilon_\lambda \in \mathcal{S}. \label{sos:multipl}
\end{align}
\end{subequations}
Then, \eqref{sol-pol:pd and rad unbnd}-\eqref{sol-pol:multipl} and the conclusion of Proposition~\ref{prop:sol-pol} hold.
\end{theorem}
\begin{proof}
\eqref{sos:pd and rad unbnd} and \eqref{sos:multipl} imply \eqref{sol-pol:pd and rad unbnd} and \eqref{sol-pol:multipl}, respectively. 
Call $x \mapsto Q(x)$ the matrix polynomial in~\eqref{sos:decrease}, so that \eqref{sos:decrease} rewrites $- Q \in \mathcal{S}_{\tu{m}}$. 
By definition of $\mathcal{S}_{\tu{m}}$, see \cite[Eq.~(9)]{chesi2010lmi}, we have that for each $x$, $Q(x) \preceq 0$, i.e., \eqref{sol-pol:decrease}.
\end{proof}

Let us comment Theorem~\ref{thm:sos}. 
Quantities $Z$ and $W$ are the known regressors; $\tilde{Z}_{\tu{c}}$, $\tilde{\mb{A}}$, $\tilde{\mb{Q}}$ are obtained from data $X_1$, $Z_0$, $V_0$; $\ell_1$, $\ell_2$ and $\epsilon_\lambda$ are design parameters; finally, $V$, $k$ and $\lambda$ are decision variables. 
Then, the blocks (1,1), (3,1) and (1,3) of the matrix in~\eqref{sos:decrease} entail products between decision variables, which make condition \eqref{sos:decrease} bilinear and the feasibility program in~\eqref{sos} not convex. 
A suboptimal strategy that is widely adopted in the sum-of-squares literature, see \cite{jarvis2005control}, is to alternately solve for $V$ with $k$ and $\lambda$ fixed, and solve for $k$ and $\lambda$ with $V$ fixed.
We illustrate this strategy in Section~\ref{sec:num ex}.

As in Section~\ref{sec:data-driv lin:main results}, when the set $\tilde{\mathcal{C}}$ is given directly in the form \eqref{set tilde C} as a matrix-ellipsoid over-approximation of a less tractable set (see  the discussion in Section~\ref{sec:C as overapprox}), a better alternative to Theorem~\ref{thm:sos} is the next corollary.
\begin{corollary}
\label{cor:assumeOnlyAandQ-pol}
Let $\tilde{\mb{A}} \succ 0$ and $\tilde{\mb{Q}} \succeq 0$ hold for the set $\tilde{\mathcal{C}} = \big\{ \big[A~~B\big] = Z^\top \colon 
(Z - \tilde{\Zc})^\top \tilde{\mb{A}} (Z - \tilde{\Zc}) \preceq \tilde{\mb{Q}}
\big\}$ in~\eqref{set tilde C}.
Given positive definite polynomials $\ell_1$, $\ell_2$ with $\ell_1$ radially unbounded and  a positive scalar $\epsilon_\lambda$, suppose there exist polynomials $V$, $k$, $\lambda$ with $V(0)=0$ and $k(0)=0$ satisfying \eqref{sos}. 
Then, \eqref{sol-pol:pd and rad unbnd}-\eqref{sol-pol:multipl} and the conclusion of Proposition~\ref{prop:sol-pol} hold.
\end{corollary}

In the next section, we obtain $\tilde{\mathcal{C}}$ as described in Section~\ref{sec:C as overapprox} and, in particular, through the optimization problem in~\eqref{opt overapprox ellips}. 
This provides a set $\tilde{\mathcal{C}}$ directly in the form \eqref{set tilde C}, so we will apply Corollary~\ref{cor:assumeOnlyAandQ-pol}.

Finally, we follow up on the comparison with~\cite{Guo2020} discussed in Section~\ref{sec:intro}. 
As the proof of Proposition~\ref{prop:sol-pol} shows, the data-based conditions \eqref{sol-pol:pd and rad unbnd}-\eqref{sol-pol:multipl} correspond naturally to enforcing model-based conditions \cite[Thm.~4.2]{khalil2002nonlinear} for all systems consistent with data. 
This makes this approach extendible to other cases such as \emph{local} asymptotic stability. 
Indeed, if we consider \cite[Thm.~4.1]{khalil2002nonlinear}, we obtain the next corollary.

\begin{corollary}
\label{cor:extens to LAS}
Let Assumption~\ref{ass:pers ext-pol} hold.
Given positive definite polynomials
$\ell_0$, $\ell_1$, $\ell_2$ and a positive scalar $c$ yielding $\mathcal{D}_c := \{ x\in \real^n \colon \ell_0(x) \le c\}$, suppose there exist polynomials $s_1$, $s_2$, $V$, $k$, $\lambda$ with $V(0)=0$ and $k(0)=0$ such that for each $x$
\begin{align}
& \hspace*{-2mm} s_1(x) \ge 0,~s_2(x) \ge 0, \label{las:s1 s2 positive}\\ 
& \hspace*{-2mm}V(x) - \ell_1(x) + s_1(x) (\ell_0(x) - c)\ge 0 \label{las:pd and rad unbnd} \\
& \hspace*{-4mm}
\resizebox{.435\textwidth}{!}{
\text{$\bmat{
\left\{ \begin{aligned}& \ell_2(x) + \dVxS \tilde{Z}_{\tu{c}}^\top \smat{Z(x) \\ W(x) k(x) } \\  & \hspace*{16mm}-s_2(x)(\ell_0(x)-c)\end{aligned} \right\}
& \star &  \star\\
\tilde{\mb{A}}^{-1/2} \smat{Z(x) \\ W(x) k(x) } & -\lambda(x) I & \star\\
\lambda(x) \tilde{\mb{Q}}^{1/2} \dVx^\top & 0 & -4 \lambda(x) I\hspace*{-1pt}
} \preceq 0$\hspace*{-2mm}}
}
\label{las:decrease} \\
& \hspace*{-2mm}\lambda(x) >0. \label{las:multipl}
\end{align}
Then, the origin of $\dot x = f_{A,B}(x)$ is locally asymptotically stable for all $(A,B) \in \tilde{\mathcal{C}}$, and in particular for $(A_\star,B_\star)$.
\end{corollary}

The proof would follow the same rationale as the proof of Proposition~\ref{prop:sol-pol}, so we sketch only the key steps to highlight that the conditions \eqref{las:s1 s2 positive}-\eqref{las:multipl} in Corollary~\ref{cor:extens to LAS} follow naturally from~\cite[Thm.~4.1]{khalil2002nonlinear}. 
\eqref{las:s1 s2 positive} and \eqref{las:pd and rad unbnd} imply that $V(x) \ge \ell_1$ for all $x \in \mathcal{D}_c$ and give \cite[Eq.~(4.2)]{khalil2002nonlinear}. 
To have \cite[Eq.~(4.4)]{khalil2002nonlinear}, we would like to impose for all $(A,B) \in \tilde{\mathcal{C}}$ that $\langle \nabla V(x) , f_{A,B}(x) \rangle \le -\ell_2(x)$ for all $x \in \mathcal{D}_c$.
This is implied by the fact that for all $x$, for all $\Upsilon$ with $\| \Upsilon \| \le 1$, $\ell_2(x) + \dVx (\tilde{Z}_{\tu{c}} + \tilde{\mb{A}}^{-1/2} \Upsilon \tilde{\mb{Q}} )^\top \smat{Z(x) \\ W(x) k(x) } -s_2(x)(\ell_0(x)-c) \le 0$. This condition is indeed obtained from~\eqref{las:s1 s2 positive}, \eqref{las:decrease}-\eqref{las:multipl} and Petersen's lemma. 
With Corollary~\ref{cor:extens to LAS}, it is immediate to write its sum-of-squares relaxation for decision variables $s_1$, $s_2$, $V$, $k$, $\lambda$ in the same way we wrote Theorem~\ref{thm:sos} with Proposition~\ref{prop:sol-pol}.

\section{Numerical examples}
\label{sec:num ex}

In this section we consider as a running example the system in~\cite[Example~14.9]{khalil2002nonlinear}, i.e.,
\begin{equation}
\label{sys-ex:pol}
\bmat{\dot x_1\\ \dot x_2} = 
\bmat{x_1^2 -x_1^3 +x_2\\0} + \bmat{0\\1}u + d.
\end{equation}
This continuous-time polynomial system can be cast in the form in~\eqref{sys-pol}. 
We consider it as such in Section~\ref{sec:num ex:pol} to illustrate Theorem~\ref{thm:sos}; we consider its linearization
\begin{equation}
\label{sys-ex:lin-ct}
\dot x= \bmat{0 & 1\\ 0 & 0 } x + \bmat{0\\ 1} u + d =: A_{\star}^{\tu{ct}} x + B_{\star}^{\tu{ct}} u + d
\end{equation}
in Section~\ref{sec:num ex:lin ct} to illustrate Theorem~\ref{thm:sol-ct}; we consider a discretization of~\eqref{sys-ex:lin-ct} for sampling time $\tau_{\tu{s}}$
\begin{equation}
\label{sys-ex:lin-dt}
x^+ = ( I + \tau_{\tu{s}} A_{\star}^{\tu{ct}} ) x +  \tau_{\tu{s}} B_{\star}^{\tu{ct}} u + d =: A_{\star}^{\tu{dt}} x + B_{\star}^{\tu{dt}} u + d
\end{equation}
in Section~\ref{sec:num ex:lin dt} to illustrate Theorem~\ref{thm:sol}.
We emphasize that these systems are used only for data generation, but the vector fields in~\eqref{sys-ex:pol}-\eqref{sys-ex:lin-dt} are not known to the data-based schemes.
For $\delta >0$, the disturbance $d$ in~\eqref{sys-ex:pol}-\eqref{sys-ex:lin-dt} is taken as $(\sqrt{\delta} \cos(2 \pi 0.4 t), \sqrt{\delta} \sin(2 \pi 0.4  t))$, where $t$ corresponds to integer multiples of $\tau_{\tu{s}}$ in discrete time.
Hence, $d$ satisfies $|d|^2 \le \delta$.
From this bound on $d$, the disturbance sequence $D$ in~\eqref{set D} satisfies then the bound $D D^\top \preceq T \delta I$. 
For the linear systems \eqref{sys-ex:lin-ct}-\eqref{sys-ex:lin-dt}, we convert the bound on $d$ into $\Delta := \sqrt{T \delta} I$ in~\eqref{set D}; for the polynomial system~\eqref{sys-ex:pol}, we retain the bound on $d$ and consider $\tilde{\mathcal{C}}$ as an ellipsoidal over-approximation of the type described in Section~\ref{sec:C as overapprox}.
We solve all numericals programs using YALMIP \cite{lofberg2004yalmip} with its sum-of-squares functionality \cite{lofberg2009sos}, MOSEK ApS and MATLAB\textsuperscript{\textregistered} R2019b.

\begin{figure}
\centerline{\includegraphics[scale=.65]{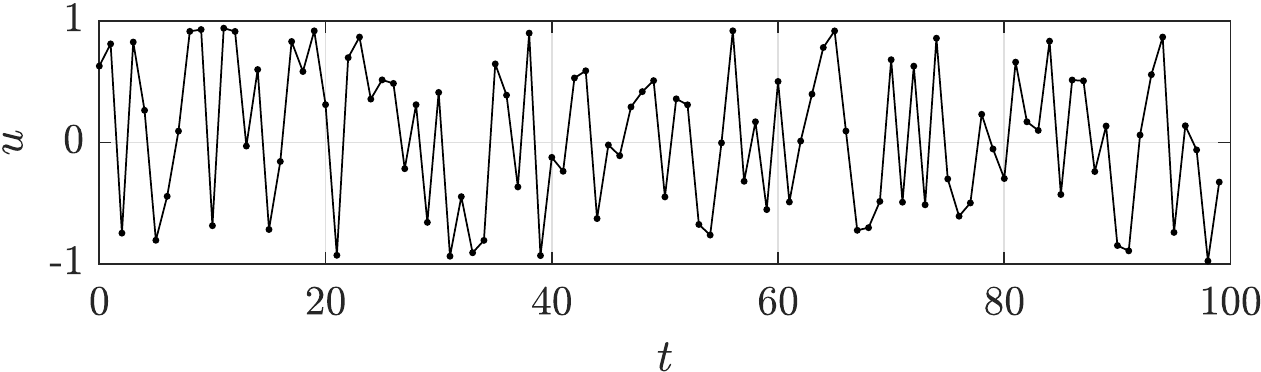}}
\centerline{\includegraphics[scale=.65]{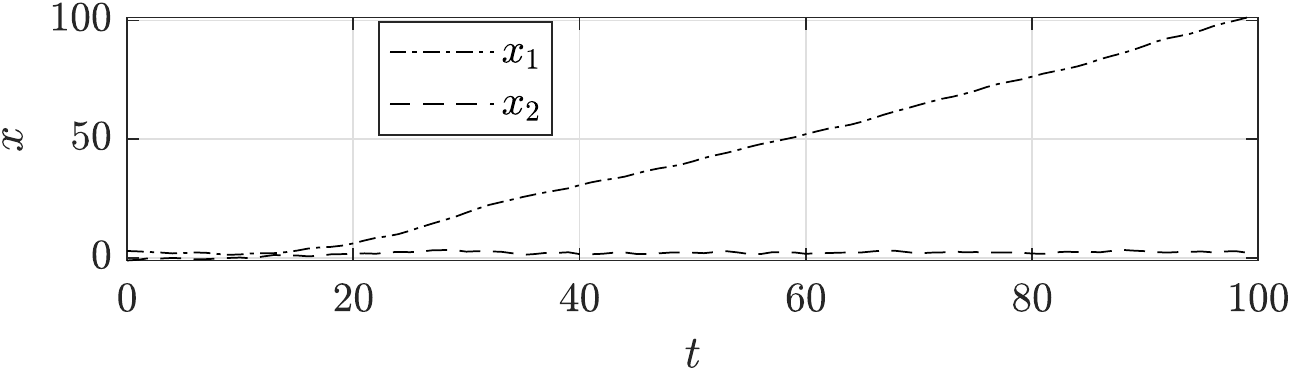}}
\caption{Experiment yielding data from~\eqref{sys-ex:lin-dt}: input and state.
Dots on the curve of $u$ indicate the discrete-time instants.
}
\label{fig:lin_dt_exp}
\end{figure}

\subsection{Linear system in discrete time}
\label{sec:num ex:lin dt}

Consider \eqref{sys-ex:lin-dt} with $\tau_{\tu{s}} = 0.5$, $T=100$ and $\delta =0.1$.
The experiment generating data $U_0$, $X_0$ and $X_1$ in~\eqref{data} is depicted in Fig.~\ref{fig:lin_dt_exp}. 
A uniform random variable in $[-1,1]$ is used as input $u$. 
Matrices $X_0$ and $U_0$ satisfy Assumption~\ref{ass:pers exc}.
Using the semidefinite program in Theorem~\ref{thm:sol}, a controller $K = \big[-0.1521~-1.3475\big]$ is designed, whose stabilization properties are certified by a Lyapunov function
$x^\top P^{-1} x = x^\top \smat{0.0043  &  0.0115\\0.0115  &  0.1000} x$.
The resulting closed-loop solutions for $d=0$ and the level sets of this Lyapunov function are depicted in Fig.~\ref{fig:lin_dt_pp}.
 
\begin{figure}
\centerline{\includegraphics[scale=.65]{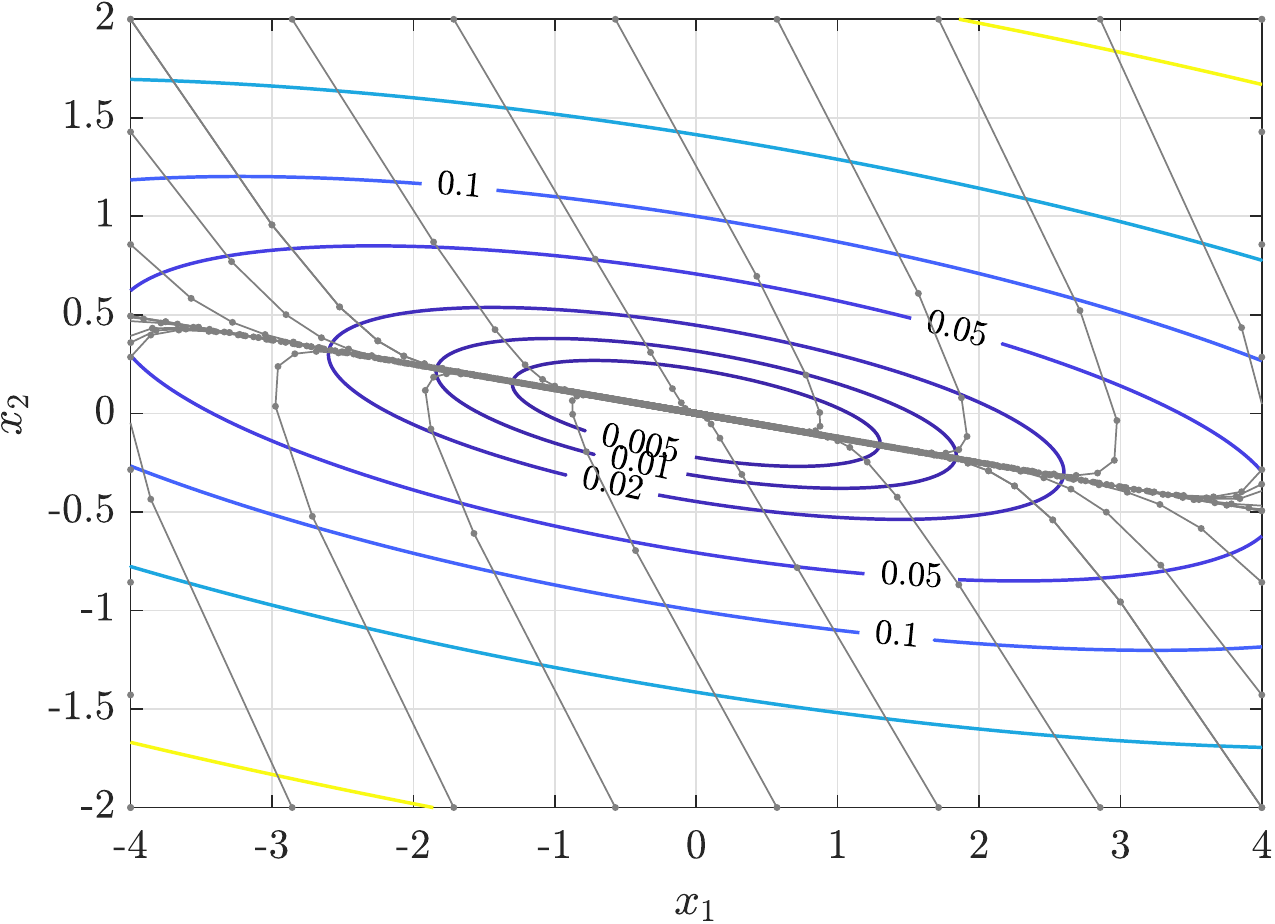}}
\caption{
Phase plot of \eqref{sys-ex:lin-dt} with $d=0$ in closed loop with the data-based controller $K$ obtained in Section~\ref{sec:num ex:lin dt}. Solutions are gray, where dots correspond to the different discrete-time instants.
The level sets of the Lyapunov function are colored and their corresponding values are annotated.
}
\label{fig:lin_dt_pp}
\end{figure} 
 
\subsection{Linear system in continuous time}
\label{sec:num ex:lin ct}

\begin{figure}
\centerline{\includegraphics[scale=.65]{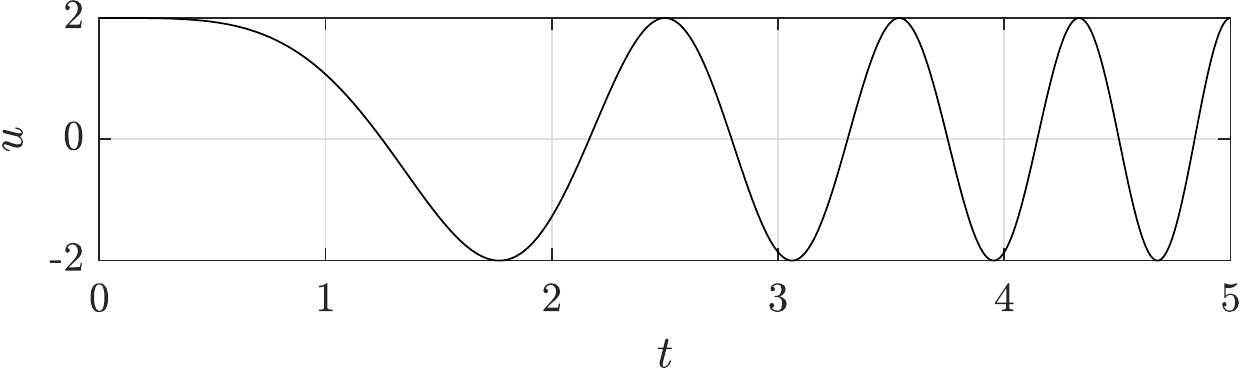}}
\centerline{\includegraphics[scale=.65]{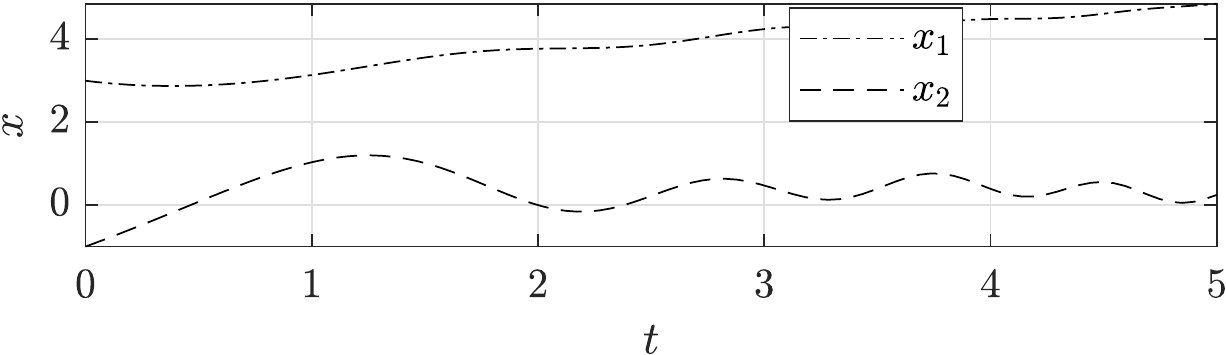}}
\caption{Experiment yielding data from~\eqref{sys-ex:lin-ct}: input and state.}
\label{fig:lin_ct_exp}
\end{figure}

Consider \eqref{sys-ex:lin-ct} with $T=100$ and $\delta =0.1$. 
The experiment generating data $U_0$, $X_0$ and $X_1$ in~\eqref{data} is depicted in Fig.~\ref{fig:lin_ct_exp}.
A sweeping sine with minimum \& maximum frequencies $0$ \& $0.8$ and amplitude $2$ is used as input $u$.
Matrices $X_0$ and $U_0$ satisfy Assumption~\ref{ass:pers exc}.
The times $t_0$, $t_1$, \dots, $t_{T-1}$ when state and state derivative are evaluated for $X_0$ and $X_1$ are uniformly spaced by $5/T=5/100$.
Using the semidefinite program in Theorem~\ref{thm:sol-ct}, a controller $K = \big[-21.4762~-9.2835\big]$ is designed, whose stabilization properties are certified by a Lyapunov function
$x^\top P^{-1} x = x^\top \smat{0.5214  &  0.1430\\0.1430  &  0.0590} x$.
The resulting closed-loop solutions for $d=0$ and the level sets of this Lyapunov function are depicted in Fig.~\ref{fig:lin_ct_pp}.

\begin{figure}
\centerline{\includegraphics[scale=.6]{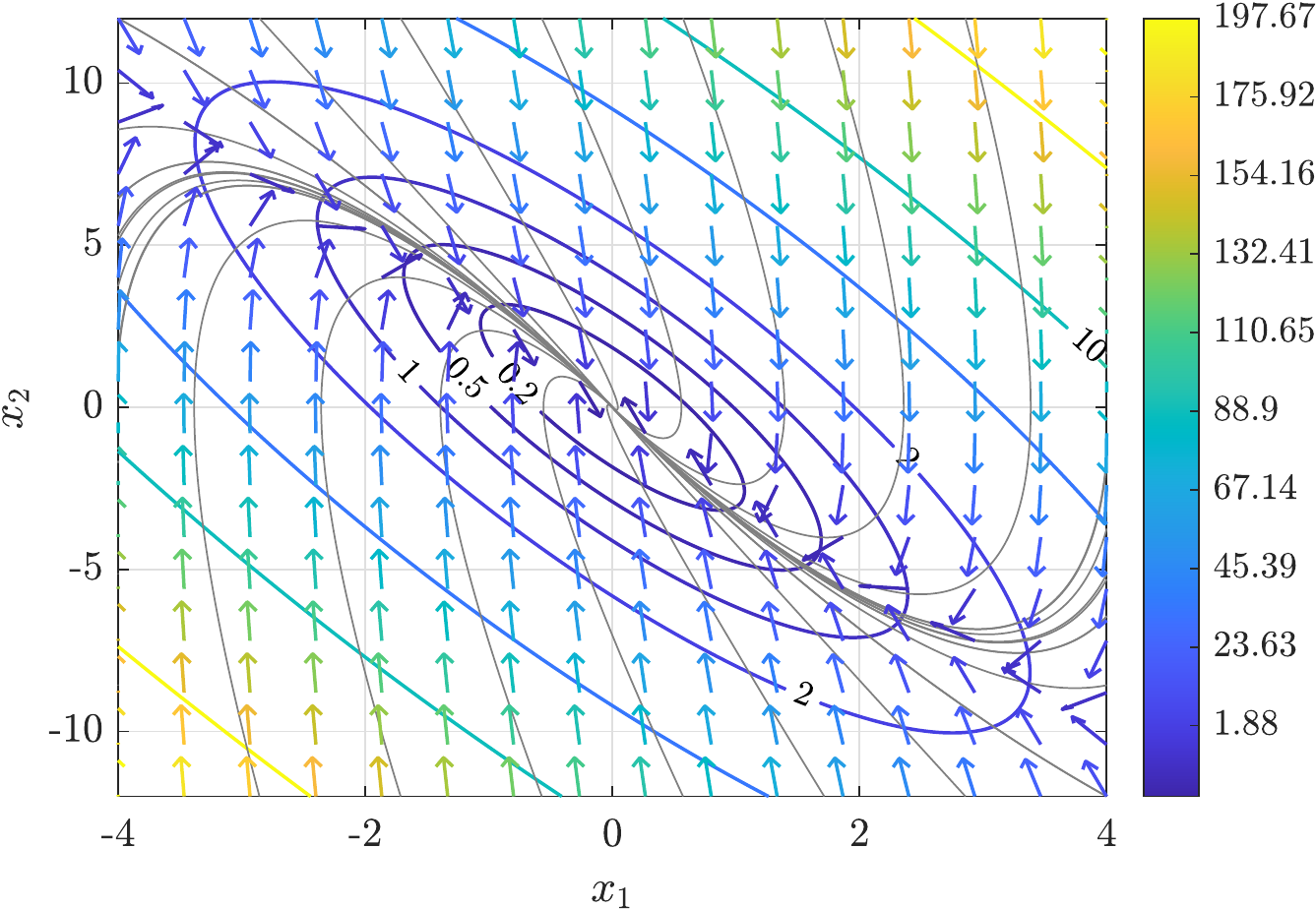}}
\caption{
Phase plot of \eqref{sys-ex:lin-ct} with $d=0$ in closed loop with the data-based controller $K$ obtained in Section~\ref{sec:num ex:lin ct}. Solutions are gray. 
The arrows represent the closed-loop vector field, and their color indicates their actual magnitude as in the right color bar.
The level sets of the Lyapunov function are colored and their corresponding values are annotated.
}
\label{fig:lin_ct_pp}
\end{figure}

\subsection{Polynomial system}
\label{sec:num ex:pol}

\begin{figure}
\centerline{\includegraphics[scale=.65]{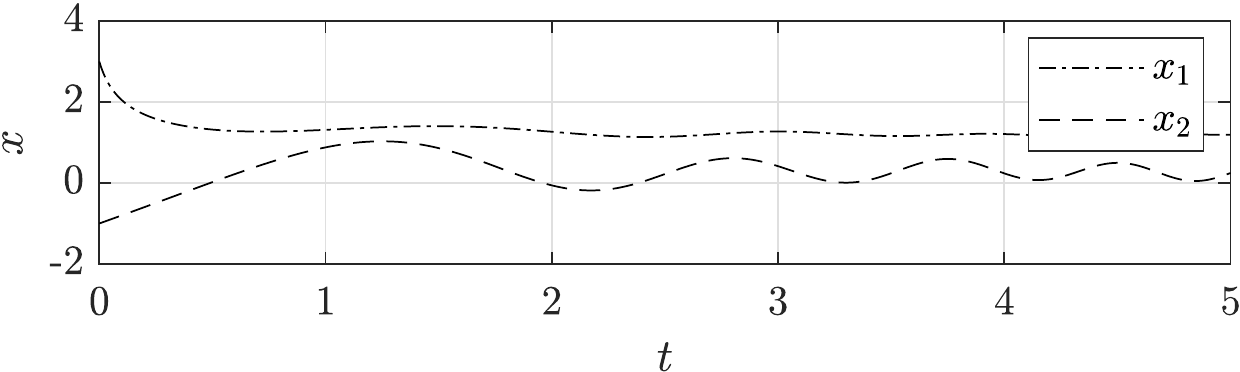}}
\caption{Experiment yielding data from~\eqref{sys-ex:pol}: state. The input is the same as in Fig.~\ref{fig:lin_ct_exp}.}
\label{fig:sos_exp}
\end{figure}

Consider \eqref{sys-ex:pol} with $T=1000$ and $\delta=0.01$. 
Input and disturbance of the experiment generating data $V_0$, $Z_0$ and $X_1$ in~\eqref{data-pol} take the same form as in Section~\ref{sec:num ex:lin ct} with times $t_0$, $t_1$, \dots, $t_{T-1}$ uniformly spaced by $5/T=5/1000$. 
Since \eqref{sys-ex:pol} is now nonlinear unlike \eqref{sys-ex:lin-ct}, these input and disturbance result in a different state evolution $x$, which is reported in Fig.~\ref{fig:sos_exp}.

Whereas the setup of the semidefinite programs in Theorems~\ref{thm:sol} and \ref{thm:sol-ct} is quite straightforward, the setup of the sum-of-squares program from Theorem~\ref{thm:sos} is less so,  and we illustrate its most relevant aspects.\smallskip

\begin{enumerate}[label=\textit{\arabic*)},left=0pt,noitemsep,nosep,wide]
\item The \emph{selection of the regressors} $Z$ and $W$ in~\eqref{sys-pol} is a key step.
On one hand, the system is unknown and some of the monomials in the regressors may not appear in the ``true'' vector fields; on the other hand, the more the monomials and the associated unknown coefficients are, the larger the uncertainty typically is in such coefficients and a too large uncertainty affects feasibility of the sum-of-square program in a critical way. 
Therefore, a parsimonious number of monomials is desirable; which monomials should be taken can be determined by trial and error by solving the program with different selections of regressors.
We select here $Z(x):=(x_2,x_1^2,x_2^2,x_1^3,x_2^3)$ and $W(x):=1$, which determine from \eqref{sys-pol} and \eqref{sys-ex:pol}
\begin{equation*}
\big[A_\star~|~B_\star\big]
:=
\left[
\begin{array}{c|c}
\begin{smallmatrix}
1 & 1 & 0 & -1 & 0\\
0 & 0 & 0 &  0 & 0
\end{smallmatrix}
& 
\begin{smallmatrix}
0\\
1
\end{smallmatrix}
\end{array}
\right].
\end{equation*}

\item The \emph{numerical experiment} is likewise important. Intuitively, the richer the data, the better; for this reason we selected as input a sweeping sine. 
Moreover, when the set $\tilde{\mathcal{C}}$ is obtained as an ellipsoidal over-approximation as described in Section~\ref{sec:C as overapprox}, the size of $\tilde{\mathcal{C}}$, the associated uncertainty, and the degree of robustness required in the design of the controller all decrease with $T$, in general; hence, more data points $T$ enlarge the feasibility set of the sum-of-squares program in~\eqref{sos}. 
We obtain the ellipsoidal over-approximation $\tilde{\mathcal{C}}$ by solving the optimization problem~\eqref{opt overapprox ellips} with \eqref{gamma_i,beta_i,alpha_i} and, for $i = 0, \dots , T-1$,
\begin{equation*}
\varkappa^\circ_i := \dot{x}(t_i), \varkappa_i:=Z(x(t_i)), \upsilon_i :=W(x(t_i)) u(t_i).
\end{equation*}
$\tilde{\mathcal{C}}$ is then defined by the  matrices $\tilde{\mb{A}}$, $\tilde{\mb{B}}$, $\tilde{\mb{C}} := \tilde{\mb{B}}^\top \tilde{\mb{A}}^{-1} \tilde{\mb{B}} - I$ returned by~\eqref{opt overapprox ellips} or, alternatively, by $\tilde{\mb{Q}}=I$ and $\tilde Z_{\tu{c}}:=- \tilde{\mb{A}}^{-1} \tilde{\mb{B}}$.
Matrix $\tilde Z_{\tu{c}}$ is especially relevant as the center of the ellipsoid $\tilde{\mathcal{C}}$. 
For the experiment in Fig.~\ref{fig:sos_exp}, we obtain 
\begin{align*}
\tilde{Z}_{\tu{c}}^\top =
\left[
\begin{array}{c|c}
\begin{smallmatrix}
0.9569  &  1.0243  &  0.0000 & -1.0084  & -0.0627\\
-0.0160 &  0.0146  & -0.0336 & -0.0037  &  0.0334
\end{smallmatrix} & 
\begin{smallmatrix}
0.0009\\
1.0101
\end{smallmatrix}
\end{array}
\right]
\end{align*}
which should be compared against $\big[A_\star~|~B_\star\big]$.\smallskip

\item To solve the sum-of-square program of Theorem~\ref{thm:sos}, we commented after it that we adopt the common practice of \emph{solving alternately two sum-of-square programs}. 
Specifically, we first solve \eqref{sos:decrease} and \eqref{sos:multipl} with respect to the controller $k$ and multiplier $\lambda$ with fixed Lyapunov function $V$; with the returned controller and multiplier, we solve \eqref{sos:pd and rad unbnd} and \eqref{sos:decrease} with respect to $V$ with fixed $k$ and $\lambda$.
To start up this procedure, we need an initial guess for $V$. 
In keeping with the data-based approach, we use the quadratic Lyapunov function that Theorem~\ref{thm:sol-ct} returns for the linearized system with same disturbance level, in this case $x^\top \smat{0.0278 & 0.0127\\ 0.0127 & 0.0216} x$. 
(This correspond to an experiment with small signals in a neighborhood of the origin, so that the linear approximation is trustworthy; the theoretical legitimacy of such an initial guess is based on \cite[Thm.~6]{depersis2020tac}.) The initialization of $V$ is all the more important whenever the feasibility set is small.
In the specific example, we run 15 iterations of this procedure (solving a total of 30 sum-of-squares programs).\smallskip

\item Finally, we mention two aspects regarding the solution of the  two alternate programs above. 
An important aspect for feasibility of each of those is the \emph{selection of the minimum \& maximum degrees of polynomials}, as lucidly explained in \cite[Appendix]{tan2006phdthesis}. 
In this example we select the minimum \& maximum degrees for $V$, $k$, $\lambda$ as respectively 2 \& 4, 1 \& 3, 0 \& 4. 
A minor aspect is that we can take in \eqref{sos} the \emph{parameter $\ell_2$ as decision variable} for greater flexibility since $\ell_2$ appears linearly anyway; when we solve for $k$ and $\lambda$, we also solve for $\ell_2$ and capture that it needs to be positive definite by imposing $\ell_2 \in \mathcal{S}$ (minimum \& maximum degrees equal to 2 \& 4).%
\end{enumerate}

With this procedure and design parameters $\ell_1(x):= 10^{-3}(x_1^2+x_2^2)$ and $\epsilon_\lambda:=10^{-3}$, the obtained $V$, $k$, $\lambda$ are in the next table; the corresponding closed-loop solutions for $d=0$ and the level sets of $V$ are depicted in Fig.~\ref{fig:sos_pp}.\newline
\scriptsize%
\resizebox{.48\textwidth}{!}{
\(
\begin{array}{ll}
\toprule
\text{Qty} & \text{Expression} \\
\midrule
\begin{aligned} V\\ \phantom{V} \\ \phantom{V} \\ \phantom{V} \end{aligned} &
\begin{aligned}
& 4.0698 x_1^2+4.3023 x_1 x_2+3.5364 x_2^2+0.003475 x_1^3\\
& \,+0.02465 x_1^2 x_2-0.01500 x_1 x_2^2+0.001575 x_2^3\\
& \,+0.008769 x_1^4+0.003686 x_1^3 x_2+0.01263 x_1^2 x_2^2\\
& \,+0.0006249 x_1 x_2^3+0.02279 x_2^4
\end{aligned} \\
\begin{aligned} k\\ \phantom{V} \end{aligned} & 
\begin{aligned}
& -1.0291 x_1-1.7292 x_2-0.8793 x_1^2+0.2927 x_1 x_2-0.07565 x_2^2\\
& \,-0.5511 x_1^3-1.6307 x_1^2 x_2+0.08336 x_1 x_2^2-2.5235 x_2^3\\
\end{aligned} \\
\begin{aligned} \lambda\\ \phantom{V} \end{aligned} & 
\begin{aligned}
& 0.04905-0.006151 x_1+0.002003 x_2+0.1106 x_1^2 \\
& \,+ 0.004398 x_1 x_2 +0.1123 x_2^2 
\end{aligned} \\
\bottomrule
\end{array}
\)
}
\normalsize

\begin{figure}
\centerline{\includegraphics[scale=.6]{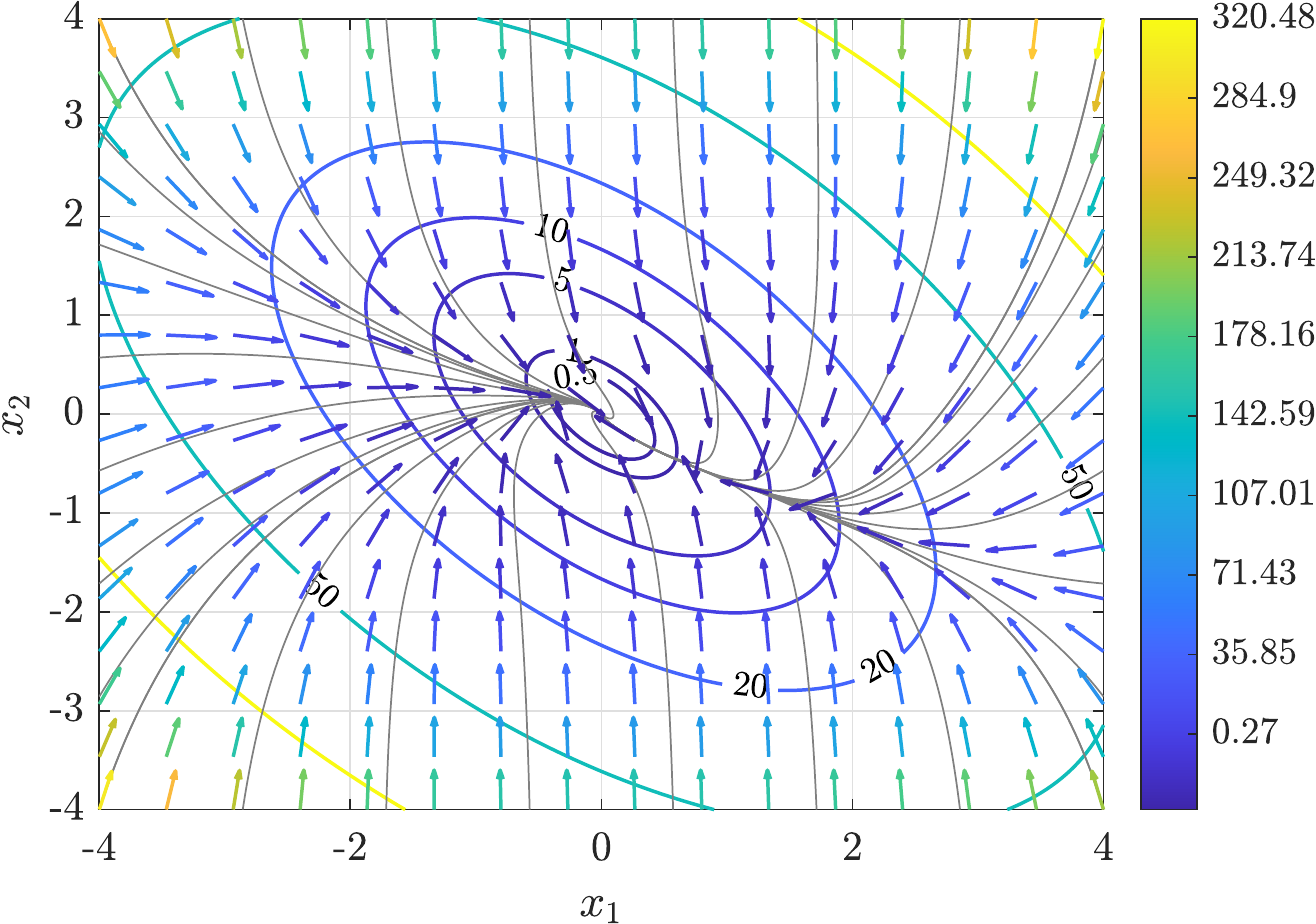}}
\caption{
Phase plot of \eqref{sys-ex:pol} with $d=0$ in closed loop with the data-based controller $x\mapsto k(x)$ obtained in Section~\ref{sec:num ex:pol}.
See the caption of Fig.~\ref{fig:lin_ct_pp} for the interpretation of this figure.}
\label{fig:sos_pp}
\end{figure}

\bibliographystyle{plain}
\bibliography{pubs-petersen}           

\appendix
\section*{Appendix}

\section{Proof of Fact~\ref{fact:petersen-ext}}

To give the proof, we need some auxiliary results. The first one is in the next fact.
\begin{fact}[{\cite[Lemma~A.4]{petersen1986riccati}}]
\label{fact:existence of multiplier}
Consider matrices $\mb{A}$, $\mb{C}$, $\mb{B}$ in $\real^{r \times r}$ with $\mb{A}=\mb{A}^\top \succeq 0$, $\mb{C}=\mb{C}^\top \succeq 0$ and $\mb{B}=\mb{B}^\top \prec 0$. Suppose further that
\begin{equation}
\label{quadratic relation}
(w^\top \mb{B} w)^2 - 4 w^\top \mb{A} w w^\top \mb{C} w > 0 \text{ for all } w \in \real^r \backslash \{0\}.
\end{equation}
Then $\lambda^2 \mb{A} + \lambda \mb{B} + \mb{C} \prec 0$ for some $\lambda > 0$.
\end{fact}

By \cite[Thm.~7.2.10]{horn2013matrix}, $\overline{\mb{F}} = \overline{\mb{F}}^\top$ is positive semidefinite if and only if there exists a $s$-by-$q$ matrix $\Phi$ such that $\overline{\mb{F}} = \Phi^\top \Phi$; hence, $\mathcal{F}$ in~\eqref{set cal F} rewrites as
\begin{equation}
\mathcal{F}:=\{\mb{F} \in \real^{p \times q}\colon \mb{F}^\top \mb{F} \preceq \Phi^\top \Phi \}. \label{set cal F alt}
\end{equation}

The second auxiliary result is in essence \cite[Lemma~3.1]{petersen1987stabilization}, of which however we need a slight extension to handle, in the set $\mathcal{F}$ in~\eqref{set cal F alt}, the condition $\mb{F}^\top \mb{F} \preceq \Phi^\top \Phi$ (with positive \emph{semidefinite} bound) instead of $\mb{F}^\top \mb{F} \preceq I$, which appears in \cite[Lemma~3.1]{petersen1987stabilization}.
This is done in the next lemma, for which we present also a short proof to account for the required modification.

\begin{lemma}
\label{lemma:var elim}
For vectors $x$ in $\mathbb{R}^{p}$, $y$ in $\mathbb{R}^{q}$ and set $\mathcal{F}$ in~\eqref{set cal F alt}, $\max _{\mb{F} \in \mathcal{F}}(x^\top \mb{F} y)^2=|x|^2\ |\Phi y|^2$.
\end{lemma}

\begin{proof} By Cauchy-Schwarz's inequality and \eqref{set cal F alt},
\[
|x^\top \mb{F} y | \le |x| |\mb{F} y| = |x| \sqrt{y^\top \mb{F}^\top \mb{F} y} \le |x| \sqrt{y^\top \Phi^\top \Phi y},
\]
that is, $|x^\top \mb{F} y | \le |x| |\Phi y |$. From this relation we have that the statement is true if $x=0$ or $\Phi y =0$.
The proof is complete if, for $x \neq 0$ and $\Phi y \neq 0$, we obtain $\mb{F} \in \mathcal{F}$ such that $|x^\top \mb{F} y | = |x| |\Phi y |$, as we do in the rest of the proof. Since $x \neq 0$ and $\Phi y \neq 0$, take the specific selection
\[
\mb{F} :=x y^\top \Phi^\top \Phi / \big(|x| |\Phi y| \big).
\]
First, we show $\mb{F} \in \mathcal{F}$. Indeed,
\[
\mb{F}^\top \mb{F}
=
\frac
{\Phi^\top ( \Phi y x^\top x y^\top \Phi^\top) \Phi}
{|x|^2|\Phi y|^2} 
=
\frac
{\Phi^\top(\Phi y y^\top \Phi^\top) \Phi}
{|\Phi y|^2}
\]
and $\mb{F} \in \mathcal{F}$ because $\Phi y y^\top \Phi^\top \preceq |\Phi y|^2 I$ (for all $v \in \real^s$, 
$v^\top \big( \Phi y y^{\top} \Phi^\top \big) v \le |v^{\top} \Phi y \| y^{\top} \Phi^\top v| \le |\Phi y|^{2} |v|^2$). Second, 
\[
(x^\top \mb{F} y)^2
=\Big(x^\top \frac{x y^\top \Phi^\top \Phi}{|x| |\Phi y|} y \Big)^2
=|x|^2 |\Phi y |^2
\]
so we have also shown $|x^\top \mb{F} y | = |x| |\Phi y |$.
\end{proof}

With Fact~\ref{fact:existence of multiplier} and Lemma~\ref{lemma:var elim}, we can prove Fact~\ref{fact:petersen-ext}. 
The direction \eqref{petersen-ext:multipl}$\implies$\eqref{petersen-ext:for all} is easy since for all $\mb{F} \in \mathcal{F}$ in~\eqref{set cal F alt},
\begin{align*}
0 & \succ \mb{C}+\lambda \mb{E} \mb{E}^\top +\lambda^{-1} \mb{G}^\top \Phi^\top \Phi \mb{G}\\
& \succeq \mb{C}+\lambda \mb{E} \mb{E}^\top +\lambda^{-1} \mb{G}^\top \mb{F}^\top \mb{F} \mb{G} \quad\text{(by $\lambda>0$)}\\
&  = \mb{C} + \mb{E} \mb{F} \mb{G} + \mb{G}^\top \mb{F}^\top\mb{E}^\top\\
& \hspace*{10mm} + 
(\sqrt{\lambda}  \mb{E}^\top  - \sqrt{\lambda}^{-1}  \mb{F} \mb{G})^\top(\sqrt{\lambda}  \mb{E}^\top - \sqrt{\lambda}^{-1}  \mb{F} \mb{G})\\
&  \succeq \mb{C} + \mb{E} \mb{F} \mb{G} + \mb{G}^\top \mb{F}^\top\mb{E}^\top.
\end{align*}
We turn then to the direction \eqref{petersen-ext:for all}$\implies$\eqref{petersen-ext:multipl}. 
\eqref{petersen-ext:for all} is equivalent to the fact that for all $x \neq 0$ and for all $\mb{F} \in \mathcal{F}$, $x^\top \mb{C} x + 2 x^\top \mb{E} \mb{F} \mb{G} x < 0$ and to the fact that for all $x \neq 0$, $0 > x^\top \mb{C} x + 2 \max_{\mb{F} \in \mathcal{F}} (x^\top \mb{E} \mb{F} \mb{G} x ) = x^\top \mb{C} x + 2 \max_{\mb{F} \in \mathcal{F}} |x^\top \mb{E} \mb{F} \mb{G} x |$  because there exists a value of $\mb{F} \in \mathcal{F}$, namely $\mb{F} = 0$, that makes $x^\top \mb{E} \mb{F} \mb{G} x$ nonnegative.
Apply Lemma~\ref{lemma:var elim} and obtain that for all $x \neq 0$, $x^\top \mb{C} x + 2 |\mb{E}^\top x| |\Phi \mb{G} x| < 0$. 
For this to hold, we necessarily have $x^\top \mb{C} x < 0$ for all $x \neq 0$, i.e., $\mb{C} \prec 0$. 
Under $\mb{C} \prec 0$, $x^\top \mb{C} x + 2 |\mb{E}^\top x| |\Phi \mb{G} x| < 0$ for all $x \neq 0$ is equivalent to the fact that for all $x \neq 0$, $(x^\top \mb{C} x)^2 > 4 |\mb{E}^\top x|^2 |\Phi\mb{G} x|^2$. 
This relation corresponds to \eqref{quadratic relation} and the hypothesis of Fact~\ref{fact:existence of multiplier} is verified. Hence, we conclude from Fact~\ref{fact:existence of multiplier} that $\lambda^2 \mb{E} \mb{E}^\top + \lambda \mb{C} + \mb{G}^\top \Phi^\top \Phi \mb{G} \prec 0$ for some $\lambda>0$, which is equivalent to $\lambda^2 \mb{E} \mb{E}^\top + \lambda \mb{C} + \mb{G}^\top \overline{\mb{F}} \mb{G} \prec 0$ and \eqref{petersen-ext:multipl}.

\arxiv{
\section{Proof of Fact~\ref{fact:petersen-nonstrict}}
We start with some preliminary claims. We report (a version of) the classical nonstrict S-procedure.
\begin{fact}[\cite{boyd1994linear}, p.~23]
\label{fact:nonstrictSproc}
Let $T_0$ and $T_1$ be $n$-by-$n$ symmetric matrices and assume there is some $\zeta_0$ such that $\zeta_0^\top T_1 \zeta_0>0$.
\begin{center}
$\zeta^\top T_0 \zeta \ge 0$ for all $\zeta$ such that $\zeta^\top T_1 \zeta \ge 0$
\end{center}
if and only if
\begin{center}
there exists $\tau_1 \ge 0$ such that $T_0 - \tau_1 T_1 \succeq 0$.
\end{center}
\end{fact}
As noted in the proof of Fact~\ref{fact:petersen-ext}, $\overline{\mb{F}} = \Phi^\top \Phi$ for some $s$-by-$q$ matrix $\Phi$ and $\mathcal{F}$ rewrites as
\begin{equation*}
\mathcal{F} =\{\mb{F} \in \real^{p \times q}\colon \mb{F}^\top \mb{F} \preceq \Phi^\top \Phi \}. 
\end{equation*}
\eqref{petersen-nonstrict:for all} is equivalent to
\begin{subequations}
\begin{align}
& x^\top( \mb{C} + \mb{E} \mb{F} \mb{G} + \mb{G}^\top \mb{F}^\top \mb{E}^\top ) x \le 0 \nonumber \\
& \hspace*{0mm}\text{ for all } x \in \mathbb{R}^{n} \text{ and } \mb{F} \in \mathbb{R}^{p \times q} \text{ with } \mb{F}^\top \mb{F} \preceq \Phi^\top \Phi. \label{petersen-nonstrict:for all alt}
\end{align}
\eqref{petersen-nonstrict:multipl} is equivalent, by Schur complement, to
\begin{align}
& \text{there exists } \lambda > 0 \colon 
0 \succeq \bmat{
\mb{C} + \frac{1}{\lambda} \mb{G}^{\top} \Phi^\top \Phi \mb{G}  & \mb{E} \\
\mb{E}^\top & - \frac{1}{\lambda} I} \notag \\
& =
\bmat{
\mb{C}  & \mb{E} \\
\mb{E}^\top & 0}
+ \frac{1}{\lambda}
\bmat{
\mb{G}^{\top} \Phi^\top \Phi \mb{G}  & 0 \\
0 & - I}. \label{petersen-nonstrict:multipl alt}
\end{align}
\end{subequations}
We can now prove the two directions of implication.

($\Longleftarrow$) Assume that \eqref{petersen-nonstrict:multipl} holds or, equivalently, \eqref{petersen-nonstrict:multipl alt}. We have then that for some $\lambda >0$
\begingroup
\setlength\arraycolsep{1.5pt}%
\thinmuskip=0.3mu plus 1mu
\medmuskip=0.6mu plus 2mu
\thickmuskip=0.9mu plus 3mu
\begin{equation}
\label{petersen-nonstrict:multipl alt alt}
\bmat{x\\ y}^\top 
\left(
\bmat{
\mb{C}  & \mb{E} \\
\mb{E}^\top & 0}
+ \frac{1}{\lambda}
\bmat{
\mb{G}^{\top} \Phi^\top \Phi \mb{G}  & 0 \\
0 & - I}
\right)
\bmat{x\\ y}
\le 0 \quad \forall (x, y).
\end{equation}
\endgroup
Consider now arbitrary $x$ and $\mb{F}$ such that $\mb{F}^\top \mb{F} \preceq \Phi^\top \Phi$. Select $y = \mb{F} \mb{G} x$ in~\eqref{petersen-nonstrict:multipl alt alt}, and obtain that for each $x$ and $\mb{F}$ such that $\mb{F}^\top \mb{F} \preceq \Phi^\top \Phi$
\begingroup
\thinmuskip=0.3mu plus 1mu
\medmuskip=0.6mu plus 2mu
\thickmuskip=0.9mu plus 3mu
\begin{align*}
0 &  \ge 
\bmat{x\\ \mb{F} \mb{G} x}^\top 
\left(
\bmat{
\mb{C}  & \mb{E} \\
\mb{E}^\top & 0}
+ \frac{1}{\lambda}
\bmat{
\mb{G}^{\top} \Phi^\top \Phi \mb{G}  & 0 \\
0 & - I}
\right)
\bmat{x\\ \mb{F} \mb{G} x}\\
& = x^\top \big( \mb{C}  + \mb{E}  \mb{F} \mb{G} + \mb{G}^\top \mb{F}^\top \mb{E}^\top
+ \frac{1}{\lambda}  \mb{G}^\top (\Phi^\top \Phi - \mb{F}^\top \mb{F}) \mb{G} \big) x\\
& \ge x^\top  ( \mb{C}  + \mb{E}  \mb{F} \mb{G} + \mb{G}^\top \mb{F}^\top \mb{E}^\top ) x
\end{align*}
\endgroup
since $\lambda > 0$ and $ \mb{F}^\top \mb{F} \preceq \Phi^\top \Phi$.
In summary, we have shown that if \eqref{petersen-nonstrict:multipl} or, equivalently, \eqref{petersen-nonstrict:multipl alt} hold, then for each $x$ and $\mb{F}$ with $\mb{F}^\top \mb{F} \preceq \Phi^\top \Phi$, $ x^\top  ( \mb{C}  + \mb{E}  \mb{F} \mb{G} + \mb{G}^\top \mb{F}^\top \mb{E}^\top ) x \le 0$, i.e., \eqref{petersen-nonstrict:for all alt} or, equivalently, \eqref{petersen-nonstrict:for all} hold. 
Since we have shown this \emph{without} invoking $\mb{E} \neq 0$ or $\overline{\mb{F}} \succ 0$ or $\mb{G} \neq 0$, we have also proven the last statement of Fact~\ref{fact:petersen-nonstrict}.

($\Longrightarrow$) First, we show that if \eqref{petersen-nonstrict:for all alt} holds, then
\begin{align}
& x^\top \mb{C} x + x^\top \mb{E} y + y^\top \mb{E}^\top x  \le 0 \notag \\
& \hspace*{5mm} \text{for all } (x,y) \colon y^\top y \le x^\top \mb{G}^\top \Phi^\top \Phi \mb{G} x. \label{cond_x_y}
\end{align}
Consider arbitrary $x$ and $y$ that satisfy $y^\top y \le x^\top \mb{G}^\top \Phi^\top \Phi \mb{G} x$, i.e., $|y|^2 \le |\Phi \mb{G} x|^2$. 
For such $x$ and $y$ we can select $\mb{F}$ such that $y=(y_1,\dots,y_p)=\mb{F} \mb{G} x$ and $\mb{F}^\top \mb{F} \preceq \Phi^\top \Phi$. 
Indeed, if $\Phi \mb{G} x = 0$, $y = 0$ and $\mb{F}=0$ in~\eqref{petersen-nonstrict:for all alt} implies that \eqref{cond_x_y} holds; if $\Phi \mb{G} x \neq 0$, select 
\[
\mb{F} = \bmat{y_1 \frac{ x^\top \mb{G}^\top \Phi^\top \Phi}{|\Phi \mb{G} x|^2}  \\ \vdots
\\
y_p \frac{ x^\top \mb{G}^\top \Phi^\top \Phi}{|\Phi \mb{G} x|^2}},
\]
which verifies $\mb{F} \mb{G} x = y$ and
\begin{align*}
\mb{F}^\top  \mb{F} &  = \frac{y^\top y}{|\Phi \mb{G} x|^4} \Phi^\top \big( \Phi \mb{G} x x^\top \mb{G}^\top \Phi^\top \big) \Phi \\
& \preceq \frac{x^\top \mb{G}^\top \Phi^\top \Phi \mb{G} x }{|\Phi \mb{G} x|^4} \Phi^\top \big( |\Phi \mb{G} x|^2 I \big) \Phi = \Phi^\top  \Phi
\end{align*}
since $x$ and $y$ satisfy $y^\top y \le x^\top \mb{G}^\top \Phi^\top \Phi \mb{G} x$.
We have shown that for all $x$ and $y$ such that $y^\top y \le x^\top \mb{G}^\top \Phi^\top \Phi \mb{G} x$, we can select $\mb{F}$ such that $y=\mb{F} \mb{G} x$ and $\mb{F}^\top \mb{F} \preceq \Phi^\top \Phi$, so that \eqref{petersen-nonstrict:for all alt} applies and gives
\begin{align*}
0 & \ge x^\top \mb{C} x + x^\top \mb{E}  \mb{F} \mb{G} x + x^\top  \mb{G}^\top \mb{F}^\top \mb{E}^\top x \\
& = x^\top \mb{C} x + x^\top \mb{E}  y + y^\top \mb{E}^\top x,
\end{align*}
i.e., \eqref{cond_x_y} holds as we wanted to show.

Second, we apply Fact~\ref{fact:nonstrictSproc} to~\eqref{cond_x_y} or, better, to its equivalent version given for variable $z:=(x,y)$ as
\begin{equation}
\label{applySproc}
z^\top \bmat{
\mb{C} & \mb{E}\\
\mb{E}^{\top}  & 0} z \le 0 
\quad \forall z \colon
z^\top \bmat{
-\mb{G}^\top \Phi^\top \Phi \mb{G} & 0 \\
0 & I} z \le 0.
\end{equation}
We start verifying the assumption of Fact~\ref{fact:nonstrictSproc} that $ z_0^\top \smat{-\mb{G}^\top \Phi^\top \Phi \mb{G} & 0 \\
0 & I} z_0 < 0$  for some $z_0$.
Suppose not by contradiction: i.e., for all $z_0$, $z_0^\top \smat{-\mb{G}^\top \Phi^\top \Phi \mb{G} & 0 \\
0 & I} z_0 \ge 0$. 
Then, for all $x_0$, $\smat{x_0\\0}^\top \smat{-\mb{G}^\top \Phi^\top \Phi \mb{G} & 0 \\
0 & I} \smat{x_0\\0} = -|\Phi \mb{G} x_0|^2 \ge 0$; i.e., for all $x_0$, $\Phi \mb{G} x_0 = 0$. 
We now use that since $\overline{\mb{F}} \succ 0$, $\Phi$ can be taken as the square root $\overline{\mb{F}}^{1/2}$ of $\overline{\mb{F}}$, which satisfies $\overline{\mb{F}}^{1/2} \succ 0$ \cite[Thm.~7.2.6]{horn2013matrix} and is invertible. 
Hence, it must be $\mb{G} x_0 = 0$ for all $x_0$ and this is a contradiction since $\mb{G} \neq 0$ by assumption and has thus at least a nonzero element which makes $\mb{G} x_0 \neq 0$ for a suitable $x_0$ (e.g., a unit versor).
Therefore, there is $z_0$ such that $z_0^\top \smat{-\mb{G}^\top \Phi^\top \Phi \mb{G} & 0 \\
0 & I} z_0 < 0$, we can apply Fact~\ref{fact:nonstrictSproc} to~\eqref{applySproc} and, by it, there exists $\lambda \ge 0$ such that
\begingroup
\setlength\arraycolsep{.5pt}%
\thinmuskip=0.1mu plus 1mu
\medmuskip=0.2mu plus 2mu
\thickmuskip=0.3mu plus 3mu
\begin{equation*}
0 \succeq
\bmat{
\mb{C} & \mb{E}\\
\mb{E}^{\top}  & 0}
-\lambda 
\bmat{
-\mb{G}^\top \Phi^\top \Phi \mb{G} & 0 \\
0 & I} 
=
\bmat{
\mb{C} + \lambda \mb{G}^\top \Phi^\top \Phi \mb{G} & \mb{E}\\
\mb{E}^{\top}  & - \lambda I}.
\end{equation*}
\endgroup
If $\lambda = 0$, we must have $\mb{E} = 0$, which contradicts the assumption; hence, the existing $\lambda$ is positive. By replacing $\lambda>0$ with $1/\lambda>0$, the last condition proves \eqref{petersen-nonstrict:multipl alt} or, equivalently, \eqref{petersen-nonstrict:multipl}.
}

\end{document}